\documentclass[reqno,12pt,a4paper]{amsart}
\usepackage{mathrsfs}
\usepackage{amssymb}
\usepackage{amsfonts}
\usepackage{latexsym}
\usepackage{amsthm}
\usepackage{graphicx}
\usepackage{pgfplots}
\usepackage{pgf}
\usepackage{url}
\usepgfplotslibrary{groupplots}
\pgfplotsset{compat=newest}
\def\lb{\label}

\newcommand{\er}[1]{\textrm{(\ref{#1})}}

\begin{document}

%%%%%%%%%% Some definitions %%%%%%%%%%

%%%%%%%% Equations, theorems %%%%%%%%%
\renewcommand{\theequation}{\arabic{section}.\arabic{equation}}
\theoremstyle{plain}
\newtheorem{theorem}{\bf Theorem}[section]
\newtheorem{lemma}[theorem]{\bf Lemma}
\newtheorem{corollary}[theorem]{\bf Corollary}
\newtheorem{proposition}[theorem]{\bf Proposition}
\newtheorem{definition}[theorem]{\bf Definition}
\newtheorem{condition}[theorem]{\bf Condition}
\newtheorem{remark}[theorem]{\it Remark}
%\theoremstyle{remark}
%\newtheorem{remark}[theorem]{\bf Remark}

%%%%% Alphabet %%%%%
\def\a{\alpha}  \def\cA{{\mathcal A}}     \def\bA{{\bf A}}  \def\mA{{\mathscr A}}
\def\b{\beta}   \def\cB{{\mathcal B}}     \def\bB{{\bf B}}  \def\mB{{\mathscr B}}
\def\g{\gamma}  \def\cC{{\mathcal C}}     \def\bC{{\bf C}}  \def\mC{{\mathscr C}}
\def\G{\Gamma}  \def\cD{{\mathcal D}}     \def\bD{{\bf D}}  \def\mD{{\mathscr D}}
\def\d{\delta}  \def\cE{{\mathcal E}}     \def\bE{{\bf E}}  \def\mE{{\mathscr E}}
\def\D{\Delta}  \def\cF{{\mathcal F}}     \def\bF{{\bf F}}  \def\mF{{\mathscr F}}
\def\c{\chi}    \def\cG{{\mathcal G}}     \def\bG{{\bf G}}  \def\mG{{\mathscr G}}
\def\z{\zeta}   \def\cH{{\mathcal H}}     \def\bH{{\bf H}}  \def\mH{{\mathscr H}}
\def\e{\eta}    \def\cI{{\mathcal I}}     \def\bI{{\bf I}}  \def\mI{{\mathscr I}}
\def\p{\psi}    \def\cJ{{\mathcal J}}     \def\bJ{{\bf J}}  \def\mJ{{\mathscr J}}
\def\vT{\Theta} \def\cK{{\mathcal K}}     \def\bK{{\bf K}}  \def\mK{{\mathscr K}}
\def\k{\kappa}  \def\cL{{\mathcal L}}     \def\bL{{\bf L}}  \def\mL{{\mathscr L}}
\def\l{\lambda} \def\cM{{\mathcal M}}     \def\bM{{\bf M}}  \def\mM{{\mathscr M}}
\def\L{\Lambda} \def\cN{{\mathcal N}}     \def\bN{{\bf N}}  \def\mN{{\mathscr N}}
\def\m{\mu}     \def\cO{{\mathcal O}}     \def\bO{{\bf O}}  \def\mO{{\mathscr O}}
\def\n{\nu}     \def\cP{{\mathcal P}}     \def\bP{{\bf P}}  \def\mP{{\mathscr P}}
\def\r{\rho}    \def\cQ{{\mathcal Q}}     \def\bQ{{\bf Q}}  \def\mQ{{\mathscr Q}}
\def\s{\sigma}  \def\cR{{\mathcal R}}     \def\bR{{\bf R}}  \def\mR{{\mathscr R}}
\def\S{\Sigma}  \def\cS{{\mathcal S}}     \def\bS{{\bf S}}  \def\mS{{\mathscr S}}
\def\t{\tau}    \def\cT{{\mathcal T}}     \def\bT{{\bf T}}  \def\mT{{\mathscr T}}
\def\f{\phi}    \def\cU{{\mathcal U}}     \def\bU{{\bf U}}  \def\mU{{\mathscr U}}
\def\F{\Phi}    \def\cV{{\mathcal V}}     \def\bV{{\bf V}}  \def\mV{{\mathscr V}}
\def\P{\Psi}    \def\cW{{\mathcal W}}     \def\bW{{\bf W}}  \def\mW{{\mathscr W}}
\def\o{\omega}  \def\cX{{\mathcal X}}     \def\bX{{\bf X}}  \def\mX{{\mathscr X}}
\def\x{\xi}     \def\cY{{\mathcal Y}}     \def\bY{{\bf Y}}  \def\mY{{\mathscr Y}}
\def\X{\Xi}     \def\cZ{{\mathcal Z}}     \def\bZ{{\bf Z}}  \def\mZ{{\mathscr Z}}
\def\O{\Omega}
\def\th{\theta}

\newcommand{\gA}{\mathfrak{A}}
\newcommand{\gB}{\mathfrak{B}}
\newcommand{\gC}{\mathfrak{C}}
\newcommand{\gD}{\mathfrak{D}}
\newcommand{\gE}{\mathfrak{E}}
\newcommand{\gF}{\mathfrak{F}}
\newcommand{\gG}{\mathfrak{G}}
\newcommand{\gH}{\mathfrak{H}}
\newcommand{\gI}{\mathfrak{I}}
\newcommand{\gJ}{\mathfrak{J}}
\newcommand{\gK}{\mathfrak{K}}
\newcommand{\gL}{\mathfrak{L}}
\newcommand{\gM}{\mathfrak{M}}
\newcommand{\gN}{\mathfrak{N}}
\newcommand{\gO}{\mathfrak{O}}
\newcommand{\gP}{\mathfrak{P}}
\newcommand{\gQ}{\mathfrak{Q}}
\newcommand{\gR}{\mathfrak{R}}
\newcommand{\gS}{\mathfrak{S}}
\newcommand{\gT}{\mathfrak{T}}
\newcommand{\gU}{\mathfrak{U}}
\newcommand{\gV}{\mathfrak{V}}
\newcommand{\gW}{\mathfrak{W}}
\newcommand{\gX}{\mathfrak{X}}
\newcommand{\gY}{\mathfrak{Y}}
\newcommand{\gZ}{\mathfrak{Z}}

\newcommand{\gm}{\mathfrak{m}}
\newcommand{\gn}{\mathfrak{n}}
\newcommand{\gf}{\mathfrak{f}}
\newcommand{\gh}{\mathfrak{h}}
\newcommand{\mg}{\mathfrak{g}}
\newcommand{\gb}{\mathfrak{b}}
\newcommand{\gp}{\mathfrak{p}}
\newcommand{\gu}{\mathfrak{u}}
\newcommand{\ga}{\mathfrak{a}}
\newcommand{\gt}{\mathfrak{t}}
\newcommand{\gs}{\mathfrak{s}}

\def\ve{\varepsilon}   \def\vt{\vartheta}    \def\vp{\varphi}    \def\vk{\varkappa}

\def\Z{{\mathbb Z}}    \def\R{{\mathbb R}}   \def\C{{\mathbb C}}    \def\K{{\mathbb K}}
\def\T{{\mathbb T}}    \def\N{{\mathbb N}}   \def\dD{{\mathbb D}}

%%%%% Arrows %%%%%

\def\la{\leftarrow}              \def\ra{\rightarrow}            \def\Ra{\Rightarrow}
\def\ua{\uparrow}                \def\da{\downarrow}
\def\lra{\leftrightarrow}        \def\Lra{\Leftrightarrow}

%%%%% Typography %%%%%

\def\lt{\biggl}                  \def\rt{\biggr}
\def\ol{\overline}               \def\wt{\widetilde}
\def\no{\noindent}

%%%%% Math signs %%%%%

\let\ge\geqslant                 \let\le\leqslant
\def\lan{\langle}                \def\ran{\rangle}
\def\/{\over}                    \def\iy{\infty}
\def\sm{\setminus}               \def\es{\emptyset}
\def\ss{\subset}                 \def\ts{\times}
\def\pa{\partial}                \def\os{\oplus}
\def\om{\ominus}                 \def\ev{\equiv}
\def\iint{\int\!\!\!\int}        \def\iintt{\mathop{\int\!\!\int\!\!\dots\!\!\int}\limits}
\def\el2{\ell^{\,2}}             \def\1{1\!\!1}
\def\sh{\sharp}
\def\wh{\widehat}
\def\bs{\backslash}
%%%%% Math operations %%%%%

\def\sh{\mathop{\mathrm{sh}}\nolimits}
\def\Area{\mathop{\mathrm{Area}}\nolimits}
\def\arg{\mathop{\mathrm{arg}}\nolimits}
\def\const{\mathop{\mathrm{const}}\nolimits}
\def\det{\mathop{\mathrm{det}}\nolimits}
\def\diag{\mathop{\mathrm{diag}}\nolimits}
\def\diam{\mathop{\mathrm{diam}}\nolimits}
\def\dim{\mathop{\mathrm{dim}}\nolimits}
\def\dist{\mathop{\mathrm{dist}}\nolimits}
\def\Im{\mathop{\mathrm{Im}}\nolimits}
\def\Iso{\mathop{\mathrm{Iso}}\nolimits}
\def\Ker{\mathop{\mathrm{Ker}}\nolimits}
\def\Lip{\mathop{\mathrm{Lip}}\nolimits}
\def\rank{\mathop{\mathrm{rank}}\limits}
\def\Ran{\mathop{\mathrm{Ran}}\nolimits}
\def\Re{\mathop{\mathrm{Re}}\nolimits}
\def\Res{\mathop{\mathrm{Res}}\nolimits}
\def\res{\mathop{\mathrm{res}}\limits}
\def\sign{\mathop{\mathrm{sign}}\nolimits}
\def\span{\mathop{\mathrm{span}}\nolimits}
\def\supp{\mathop{\mathrm{supp}}\nolimits}
\def\Tr{\mathop{\mathrm{Tr}}\nolimits}
\def\BBox{\hspace{1mm}\vrule height6pt width5.5pt depth0pt \hspace{6pt}}
\def\as{\text{as}}
\def\all{\text{all}}
\def\where{\text{where}}
\def\Dom{\mathop{\mathrm{Dom}}\nolimits}
\def\Var{\mathop{\mathrm{Var}}\nolimits}
\def\Argmax{\mathop{\mathrm{Argmax}}\nolimits}
\def\Argmin{\mathop{\mathrm{Argmin}}\nolimits}
\def\Lip{\mathop{\mathrm{Lip}}\nolimits}
%%%%%%%%%%%%% specialities %%%%%%%%%%%%%%

\newcommand\nh[2]{\widehat{#1}\vphantom{#1}^{(#2)}}
%{{\mathop{#1}\limits^\wedge}\vphantom{#1}^{(#2)}}
\def\dia{\diamond}

\def\Oplus{\bigoplus\nolimits}

%%%%%%%%%%% End of definitions %%%%%%%%%%

%%%%% OLD OLD OLD

\def\qqq{\qquad}
\def\qq{\quad}
\let\ge\geqslant
\let\le\leqslant
\let\geq\geqslant
\let\leq\leqslant
\newcommand{\ca}{\begin{cases}}
\newcommand{\ac}{\end{cases}}
\newcommand{\ma}{\begin{pmatrix}}
\newcommand{\am}{\end{pmatrix}}
\renewcommand{\[}{\begin{equation}}
\renewcommand{\]}{\end{equation}}
\def\eq{\begin{equation}}
\def\qe{\end{equation}}
\def\[{\begin{equation}}
\def\bu{\bullet}

\newcommand{\fr}{\frac}
\newcommand{\tf}{\tfrac}

\title[Convergence of Metadynamics]
{On the Convergence of Metadynamics with Gaussian Hills}

\date{\today}
\author[Andrey Badanin]{Andrey Badanin}
\author[Olga Rogacheva]{Olga Rogacheva}
\address{Saint-Petersburg
State University, Universitetskaya nab. 7/9, St. Petersburg,
199034 Russia,
an.badanin@gmail.com,\  a.badanin@spbu.ru,\
o.rogacheva@spbu.ru}


\maketitle

\begin{abstract}
Metadynamics is a class of enhanced-sampling methods that is widely used in molecular modeling. Here, we consider metadynamics with a one-dimensional collective variable and explore the limitations of using Gaussian hills in light of existing convergence results. We reduce the corresponding evolution problem to a replicator-type differential equation and analyze its long-time behavior. For a periodic collective variable, we prove the convergence of metadynamics. However, the situation is fundamentally different when the collective variable is defined on a finite interval. In this case, the stationary solution only exists in the weak sense as a finite atomic measure. The solution to the differential equation weakly converges to it. Consequently, metadynamics is ineffective due to the absence of a clear quasi-stationary state.  A quasi-stationary, transient solution can only be captured in the "INTERVAL" framework if the free energy outside the finite interval  remains nearly constant across sufficiently large regions compared to $\sigma$. 

\end{abstract}

\section{Introduction and main results}
\setcounter{equation}{0}

\subsection{Introduction}

Metadynamics is a group of enhanced-sampling methods widely used
in molecular modeling. These methods accelerate rare events and
allow the free energy of complex physical, chemical, and biological
systems to be determined as a function of one or several slow degrees
of freedom \cite{LP2}, \cite{BBP8}, \cite{V16}. In this context,
the term "slow" applies to system movements, whose relaxation times
are long compared with other characteristic relaxation times of the
system \cite{BBP11}. As metastable states (such as reactants and products)
are usually separated by high free-energy barriers and transitions between
them are slow, these slow degrees of freedom can provide reaction coordinates,
making them of particular practical interest. A slow mode can often be
represented by an explicit or implicit, e.g. machine-learning-based,
function of the atomic coordinates $s_j,j=1,..,n$,
referred to as a collective variable~(CV),
$x=x(s_1,s_2,\ldots,s_n)$ \cite{BBP11}.
If the CV distinguishes between the relevant metastable
states and resolves the transition region between them, the free-energy
difference between the states and the barrier separating them can be estimated
in principle from the free-energy profile along the CV.

Let $x$ be a collective variable. We assume that $x$ ranges over
a bounded or unbounded interval of the real axis or is periodic.
Let $F(x)$ denote the free energy of the system as a function of $x$,
defined up to an additive constant. In units where the Boltzmann
constant $k_{\rm B}=1$, the equilibrium probability density along
the collective variable is $p_{\rm eq}(x)\propto e^{-F(x)/T}$, where
$T$ is the temperature. Consequently, $F(x)=-T\ln p_{\rm eq}(x)+C$.
For an ergodic system, the equilibrium probability density
can be estimated from the fraction of simulation time spent
in a small neighborhood of $x$. Thus, in principle, the free-energy profile
can be reconstructed from a sufficiently long equilibrium molecular-dynamics
trajectory.

The main difficulty with standard molecular dynamics is that high free-energy
barriers may trap the system in a metastable state for times much longer than
those accessible in a simulation. Various enhanced-sampling methods have been
developed to overcome this problem, including metadynamics, introduced by
Laio and Parrinello \cite{LP2}. The basic idea of metadynamics is to
construct a history-dependent bias potential by regularly adding small
local hills, conventionally Gaussian functions, centered at the current
value of the collective variable (see Fig.~\ref{Figmd}). The accumulated
hills form the bias potential $V(x,t)$. By progressively filling
the regions already visited, this bias facilitates escape from
metastable states and promotes sampling over the range of
the collective variable.

\begin{figure}
    \tiny
    \unitlength 0.70mm
    \begin{picture}(167.25,69.75)(0,0)
        \put(9.5,19.5){\line(1,0){157.75}}
        \thicklines
        \qbezier(11.5,69.75)(40.625,-7)(82.25,35.25)
        \qbezier(82.5,35.25)(99,51)(124.5,39.75)
        \qbezier(124.75,39.75)(150.625,28)(167,37.25)
        \thinlines
        \qbezier(38.375,27.563)(40.688,29.188)(43.5,35.813)
        \qbezier(49.875,23.813)(52.188,25.438)(55,32.063)
        \qbezier(60.375,26.563)(62.688,28.188)(65.5,34.813)
        \qbezier(53.063,27.563)(50.75,29.188)(47.938,35.813)
        \qbezier(64.563,23.813)(62.25,25.438)(59.438,32.063)
        \qbezier(75.063,26.563)(72.75,28.188)(69.938,34.813)
        \qbezier(43.5,35.813)(44.781,38.469)(45.688,38.5)
        \qbezier(55,32.063)(56.281,34.719)(57.188,34.75)
        \qbezier(65.5,34.813)(66.781,37.469)(67.688,37.5)
        \qbezier(47.938,35.813)(46.656,38.469)(45.75,38.5)
        \qbezier(59.438,32.063)(58.156,34.719)(57.25,34.75)
        \qbezier(69.938,34.813)(68.656,37.469)(67.75,37.5)
        \put(165.75,15.75){\makebox(0,0)[cc]{$x$}}
        \put(142.25,36.75){\makebox(0,0)[cc]{$F(x)$}}
        \put(45.75,20.25){\line(0,-1){1.25}}
        \put(57.5,20){\line(0,-1){1.25}}
        \put(69,20.25){\line(0,-1){1.25}}
        \put(44.75,15.75){\makebox(0,0)[cc]{$x_1$}}
        \put(56.25,15.75){\makebox(0,0)[cc]{$x_2$}}
        \put(69.25,15.75){\makebox(0,0)[cc]{$x_3$}}
    \end{picture}
    \linethickness{0.4pt}
    \caption{Illustration of the metadynamics framework,
    x is the collective variable, F(x) is the unknown free energy profile.
    Three successive hills are deposited, each centered at a different point:
    $x_1$, $x_2$, and $x_3$.}
    \lb{Figmd}
\end{figure}
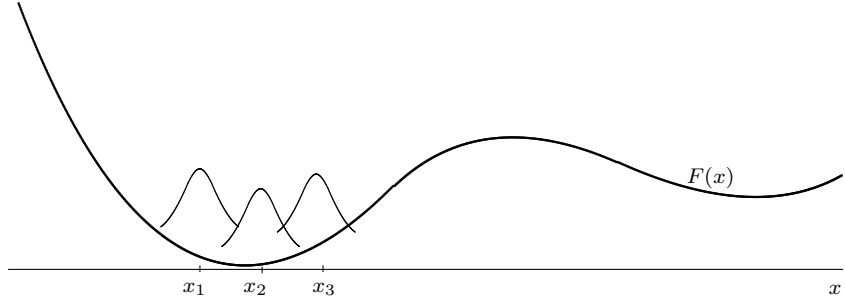

The two most common modifications of metadynamics are known as standard \cite{LP2} and well-tempered \cite{BBP8}.
In standard metadynamics, constant-height hills are deposited.
Under the assumptions required for metadynamics convergence,
a topic to be addressed subsequently, the shape of the time-averaged
bias approaches the negative free-energy profile:
\[
\lb{limit_st}
\lim_{t\to\infty}\frac{1}{t-t_0}\int_{t_0}^{t}\bigl[V(x,\tau)-C(\tau)\bigr]\,d\tau=-F(x),
\]
where $C(\tau)$ is independent of $x$, and $t_0$ is a sufficiently large
fixed point in time \cite{BLP6} \cite{CL10}. In well-tempered metadynamics,
the hill height $h$ decreases as the accumulated bias grows:
$h_{\rm current}=h_{\rm init}\exp[-\frac{V(x_i,t)}{\Delta T}]$,
where $V(x_i,t)$ is the accumulated bias at the current CV value $x_i$ and
$\Delta T$ is a tuning parameter that controls the rate at which the hill
height decreases \cite{BBP8}. Provided that convergence of well-tempered
metadynamics is established, the weighted shape of the bias converges to
the negative free-energy profile \cite{DPV14}, \cite{D16}:
\[
\lb{limit_wt}
 \lim_{t\to\infty}\Big[\frac{T+\Delta T}{\Delta T}V(x,t)-C(t)\Big]=-F(x).
\]
Thus, assuming convergence, both metadynamics modifications can be used to restore the unknown $F(x)$ based on the bias potential.

In the limit $\Delta T\to\infty$, well-tempered metadynamics formally reduces to standard metadynamics. However, this transition to the limit requires that the stronger statement \er{limit_wt} must be replaced by the weaker statement \er{limit_st} \cite{BLP6} \cite{CL10} \cite{D16}.

The relationship between accumulated bias and free energy
is founded on the rapid equilibration of the system in relation
to the evolution of the bias potential \cite{BBP11}, \cite{BLZ21}.
The extent to which this condition is met depends primarily on
the quality of the CV. Slow degrees of freedom not represented by
the CV can hinder rapid equilibration at a fixed CV, thereby preventing
free-energy reconstruction. In addition to the CV requirement, metadynamics
requires timescale separation along the CV. The equilibration of the system
 along the CV should be much faster than the change in the bias potential.
 This condition can be achieved in standard metadynamics by selecting small
 hills and a long deposition stride, as discussed in \cite{LP5}.
 In well-tempered metadynamics, decreasing the hill height leads to
  slower bias evolution. Consequently, as time approaches infinity,
  the separation of timescales automatically becomes valid provided
  the CV requirement is met \cite{DPV14}. Throughout this work, it
  is assumed that both conditions are satisfied and that the system
  has sufficiently equilibrated with respect to hill deposition.

The convergence of metadynamics has been extensively discussed in
the physical literature (see, for example, \cite{BLP6}, \cite{MP13},
and \cite{D11}). Subsequently, the theoretical basis for the convergence
 of well-tempered metadynamics was proposed by Dama, Parrinello, and
 Voth \cite{DPV14}. They derived sufficient conditions imposed on the
 hill function $G(x,x')$ under which the stochastic evolution of the bias
 potential can be described in the long-time limit by an ordinary differential
  equation (ODE):
\begin{equation}
    \frac{\partial\widetilde V(x,\tau)}{\partial\tau}
    =
    \int
    \Gamma_0(x,x')
    p_b(x';\widetilde V(\tau))\,dx',
    \label{dpv}
\end{equation}
where $\Gamma_0(x,x')$ is the zero-mean hill
function with its center in $x'$, $\wt V(x, t)$
is the zero-mean bias potential, $\tau$ is the internal metadynamics time,
and $p_b(x';\widetilde V(\tau))$
is {\it the biased equilibrium distribution} in potential $F(x)+\wt V(x)$.
The stationary state of this equation yields the conventional
free-energy reconstruction formula  \er{limit_wt}. According to
the original paper, the conditions are as follows:
1) $\int G(x, x')dx/\int dx$ must exists and be positive,
the region must be finite in size, and $\Delta T$ must be finite;
2) the equation  $\int G(x, x') p_b(x')dx' = C$ must have at least one positive,
 normalized, physically possible solution $p_b(x)$ for some constant C,
 where $p_b(x)$ is a distribution in the total $F(x) + V(x)$ potential
 3) $G(x, x')$ must be positive semidefinite.
 In most modern metadynamics implementations, $G(x, x')$
 is represented by a Gaussian function \cite{PLMD9}, \cite{CLV13}.
 There are only a few exceptions such as
 \cite{CPMD26}, \cite{T23}, \cite{MP13}, \cite{DHSV15}.
 However, the Gaussian generally cannot satisfy the second condition.
 The problem associated with using Gaussian hills was recognized long
 ago and is typically referred to as the metadynamics edge effects
 \cite{B20}. A few remedies were suggested to help level it out to some extent.
 According to the scheme "INTERVAL" \cite{BL12}
 (the name was introduced in the PLUMED package \cite{PLMD14}),
 Gaussian hills are deposited along the entire CV, but truncated
 at the boundaries of the selected interval and continue smoothly
 as constants outside it. Crespo et al. \cite{CL10} proposed a boundary
 correction in which each Gaussian deposited near a boundary
 is supplemented by its mirror image across the boundary.
 Finally, McGovern and de Pablo \cite{MP13} introduced a
 new shape of hill that satisfies at least the second requirement
 postulated by Dama, Parrinello, and Voth \cite{DPV14}. This new shape
 was later modified and formed the basis of a new metadynamics modification
 called metabasin metadynamics \cite{DHSV15}.

In light of research on metadynamics convergence,
our work aims to closely examine the limits of applying
Gaussian functions as hills. We consider three cases:

1. The collective variable is periodic.

2. The collective variable ranges over a finite interval.

3. The collective variable is confined using an ``INTERVAL'' scheme \cite{BL12}.

\subsection{Periodic problem}
The differential equation \er{depc}
that describes the change of potential in metadynamics
can be derived in several ways. Here, we present briefly
a derivation following the approach outlined
in the paper by Dama, Parrinello, and Voth \cite{DPV14}.
The basic equation of the well tempered metadynamics has the form
\[
\lb{ie}
V_{n+1}(x)=V_n(x)+\cG(x,x_{n+1})e^{-{V_n(x_{n+1})\/\D T}},\qq n=0,1,2,...,\qq
V_0(x)=0,
\]
where $x$ is the collective variable, $\D T>0$ is a parameter,
$\cG(x,s)$ is a hill-shaped function with a center at $s$
that is added to the potential,
the sequence $x_n$ represents the centers of hils added at
successive points in time $t_n$, $n=1,2,...$

The preceding equation describes the following process.
The system is to be in a state corresponding to the value $x_n$ of
the collective variable at time $t_n$, with an added potential $V_n(x)$.
The total potential is then given by $F(x) + V_n(x)$.  At time $t_{n+1}$,
the system is in a state
corresponding to the value $x_{n+1}$ of the collective variable.
A hill $\cG$ with its peak at $x_{n+1}$ is added to the potential
$V_n(x)$; the height of this hill is proportional to
$e^{-V_n(x_{n+1})/\Delta T}$ and is determined by the value of the potential
$V_n$ at $x_{n+1}$. That is to say, the height of the added hill
is lower for higher values of the potential. It is important to note that,
in the limit $\D T\to+\iy$, the corresponding equation
for the standard metadynamics is obtained.

In addition, at each stage of the process,
the mean value will be subtracted from the biased potential $V_n$.

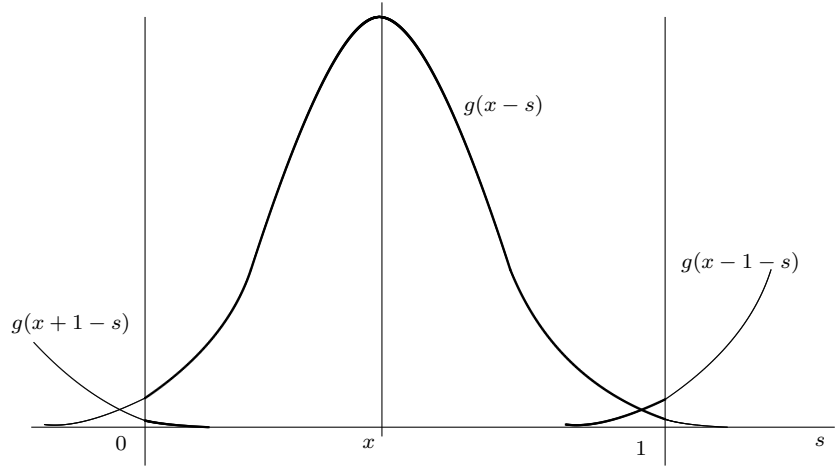
\begin{figure}
\tiny
\unitlength 0.70mm
\linethickness{0.4pt}
\begin{picture}(155.05,95.5)(0,0)
\put(3.25,15.25){\line(1,0){151.8}}
\put(24.75,8){\line(0,1){84.75}}
\put(123,8){\line(0,1){84.75}}
\put(69.5,95.5){\line(0,-1){82}}
\thicklines
\qbezier(69,92.75)(60.125,92.625)(44.75,45)
\qbezier(69,92.75)(78.375,92.625)(93.75,45)
\put(20.25,12.25){\makebox(0,0)[cc]{$0$}}
\put(118.5,11.25){\makebox(0,0)[cc]{$1$}}
\put(67,12.25){\makebox(0,0)[cc]{$x$}}
\put(152.25,12.5){\makebox(0,0)[cc]{$s$}}
\put(92.5,76){\makebox(0,0)[cc]{$g(x-s)$}}
\put(10.75,34.75){\makebox(0,0)[cc]{$g(x+1-s)$}}
\put(137.25,46.25){\makebox(0,0)[cc]{$g(x-1-s)$}}
\qbezier(93.75,45)(102.125,24.125)(123,16.75)
\qbezier(24.75,16.5)(28.375,15.625)(36.75,15.25)
\qbezier(44.75,45)(40.125,31.125)(24.75,20.75)
\qbezier(123,20.5)(111.25,14.875)(104.25,15.75)
\thinlines
\qbezier(123,16.75)(126.375,15.625)(134.75,15.25)
\qbezier(143.,45)(138.625,31.125)(123,20.5)
\qbezier(24.75,20.75)(12.75,14.875)(5.75,15.75)
\qbezier(3.75,31.25)(13.75,20.5)(24.75,16.5)
\end{picture}
\caption{Gaussians for the periodic case}
\lb{ssp}
\end{figure}

In the 1-periodic case,
the mean value is the integral over the period of the collective variable
$\ol{V_n}=\int_0^1V_n(x)dx$.
The added hill-shaped function $\cG(x,s)$ has the form
\[
\lb{defG}
\cG(x,s)=hG(x-s),\qq G(x)=\sum_{n\in\Z}g(x-n),\qq x,s\in[0,1],
\]
where the gaussian $g$ is given by
\[
\lb{gaussp}
g(x)={e^{-{x^2\/2\s^2}}\/\sqrt{2\pi}\s},\qq
\s>0,
\]
see~Fig.~\ref{ssp},
$h>0$ is the initial height of the Gaussian.
Then Eq.~\er{ie} gives the following equations for $\ol V_n$
and for the rest part
$\wt V_n(x)=V_n(x)-\ol V_n$ of the potential:
$$
\ol V_{n+1}=\ol V_n+he^{-{\ol V_n\/\D T}}e^{-{\wt V_n(x_{n+1})\/\D T}},
$$
\[
\lb{ie2p}
\wt V_{n+1}(x)=\wt V_n(x)
+he^{-{\ol V_n\/\D T}}e^{-{\wt V_n(x_{n+1})\/\D T}}
\big(G(x-x_{n+1})-1\big),
\]
we used the identity
$\int_0^1 G(x-s)du=\int_\R g(x-s)dr=1$.

Introduce the new time scale
$\tau(t_n)=h\sum_{j=0}^{n-1}e^{-{\ol V_j\/\D T}}$.
Assuming that the sequence $|\wt V_n|$ is bounded, we have
$a_1\le|\ol V_{n+1}-\ol V_n|e^{\ol V_n\/\D T}\le a_2$, for all $n\in\N$,
and for some $a_2>a_1>0$. This yields
$\tf{b_1}{n}\le e^{-{\ol V_n\/\D T}}\le \tf{b_2}{n}$, for all $n\in\N$,
and for some $b_2>b_1>0$.
Thus, when the time is increasing, the value $e^{-{\ol V_n\/\D T}}$
is decreasing, and many hills are added to accumulate the same increment.
The variable $\tau$ is the time for
which the bias potential is increased at a constant effective rate.
It satisfies $c_1\log t\le\tau(t)\le c_2\log t$ for some $c_2>c_1>0$

Eq.~\er{ie2p} takes the form
\[
\lb{ie3p}
{\wt V_{n+1}(x)-\wt V_n(x)\/\tau(t_{n+1})-\tau(t_n)}
=e^{-{\wt V_n(x_{n+1})\/\D T}}
\big(G(x-x_{n+1})-1\big).
\]

Introduce the function $V(x,\tau)$
so that $V(x,\tau(t_n))=V_n(x)$ for all $x\in[0,1]$ and for all
$n\in\N$ large enough.
We introduce the functions $\wt V(x,\tau)$ and $\ol V(\tau)$ by
$$
\ol V(\tau)=\int_0^1V(x,\tau)dx,\qq \wt V(x,\tau)=V(x,\tau)-\ol V(\tau).
$$
After a sufficient amount of time,
the number of visits to each state will be distributed according
to the Boltzmann law: the biased equilibrium  distribution $ p_b(x,\tau)$
at the time $\tau$ will be equal to
\[
\lb{rhoZ}
 p_b(x,\tau)=\Big(\int_0^1e^{-{F(s)+V(s,\tau)\/T}}ds\Big)^{-1}
 e^{-{F(x)+V(x,\tau)\/T}}
={ e^{-{F(x)+\wt V(x,\tau)\/T}}\/Z_b(\tau)},\qq
Z_b(\tau)=\int_0^1e^{-{F(s)+\wt V(s,\tau)\/T}}ds,
\]
$T$ is the temperature of the system.

Asymptotically, for large $n$,
Eq.~\er{ie3p} reduces to the differential equation
\[
\lb{depc}
{d\wt V(x,\tau )\/d\tau }={1\/Z_b(\tau)}
\int_0^1\big(G(x-s)-1\big)
e^{-\a\wt V(s,\tau )-\b F(s)}ds,
\]
where $(x,\tau)\in[0,1]\ts\R_+$,
\[
\lb{alhabeta}
\a={1\/T}+{1\/\D T},\qq\b={1\/T}.
\]

Rewrite Eq.~\er{depc} in the form
\[
\lb{eqwtVfiprpc}
{d\wt V(x,\tau )\/d\tau }={Z_w(\tau)\/Z_b(\tau)}
\Big(\int_0^1G(x-s)p_w(s,\tau)ds-1\Big),
\]
where $p_w$ is the {\it current tempering-reweighted distribution}
\[
\lb{defgpZ1pc}
p_w(s,\tau)={e^{-\a\wt V(s,\tau)-\b F(s)}\/Z_w(\tau)},
\qq
Z_w(\tau)=\int_0^1e^{-\a\wt V(s,\tau)-\b F(s)}ds.
\]
Instead $\tau$, we introduce the new time variable
$$
\theta(\tau)=\int_0^\tau {Z_w(\tau_1)\/Z_b(\tau_1)}d\tau_1
$$
and we denote
$\cV(x,\theta)=\wt V(x,\tau)$.
Then
$$
{d\theta\/d\tau}={Z_w(\tau)\/Z_b(\tau)}=\int_0^1e^{-{\wt V(s,\tau)\/\D T}}p_b(s)ds.
$$
This is an averaging of the spatially dependent part
of the tempering factor over the current biased distribution.
If we assume $|\wt V|\le C$ for some $C>0$, then
$e^{-{C\/\D T}}\le{Z_w\/Z_b}\le e^{{C\/\D T}}$, therefore,
$c_1\tau\le\theta(\tau)\le c_2(\tau)$ for some $c_2>c_1>0$.
The times $\theta$ and $\tau$ are equivalent,
the variable
$\theta(\tau)$ ranges over all $\R_+$, when $\tau$ ranges $\R_+$.

Eq.~\er{eqwtVfiprpc}
takes the form
\[
\lb{eqcVfiprpc}
{d\cV(x,\theta )\/d\theta }=G\star \gp(x,\theta)-1,
\]
where
\[
\lb{defgpg*fpc}
\gp(x,\theta)=p_w(x,\tau),\qq
G\star f(x)=\int_0^1G(x-s)f(s)ds.
\]

We consider Eq.~\er{depc},
where $(x,\tau)\in\T\ts\R_+$, $F\in C(\T),\T=\R/\Z$ is a given positive function,
$h,\a,\b$ are given positive numbers.
We search for the solution $\wt V(\cdot,\tau),\tau\in\R_+$,
in the class
$$
C_0(\T)=\{f\in C(\T):\int_0^1f(x)dx=0\}
$$
of continuous in $x$ on the circle $\T$
functions with zero mean value.
The following theorem shows that in the periodic case,
w.t.metadynamics converges to a good distribution, see~Fig.\ref{periodic}.

\begin{theorem}
\lb{Thper}
Let $F\in C(\T)$ and let $\wt V_0\in C_0(\T)$. Then
Eq.~\er{depc} has the unique solution $\wt V\in C^1(\R_+,C_0(\T))$ such that
$\wt V(x,0)=\wt V_0(x),x\in\T$.
Moreover,
\[
\wt V(x,\tau)\to {\ol F-F(x)\/{T\/\D T}+1},\qq \ol F=\int_0^1F(x)dx,
\]
as $\tau\to\iy$, uniformly on $x\in\T$.
\end{theorem}

\begin{figure}[htbp]
    \centering
    \input{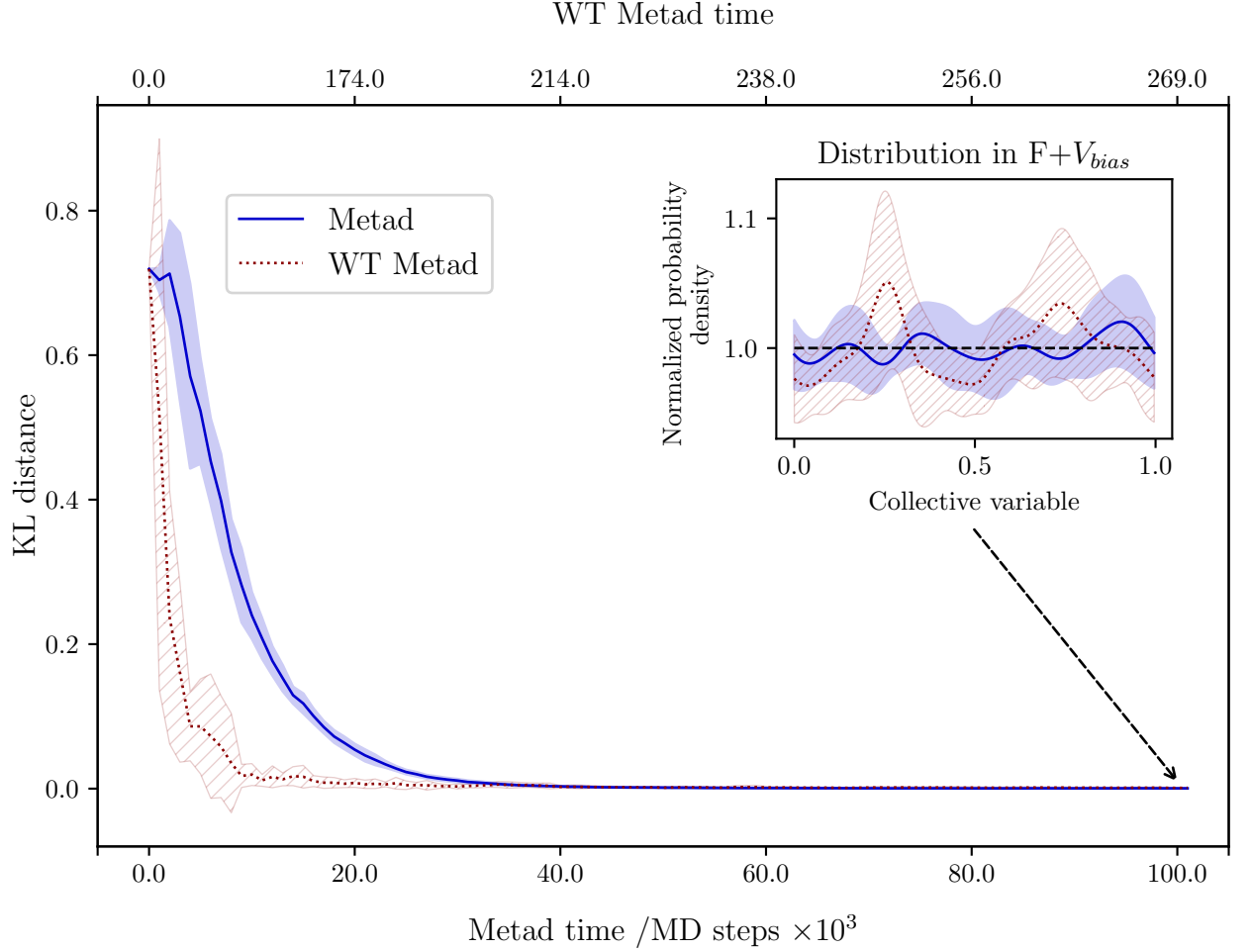}
    \caption{The results of the numerical simulations on the periodic interval [0, 1]. The original free energy function, $F(x)$, is defined on this interval as $5 \cos^2(2\pi(x - 0.5)) + \cos^2(\pi(x - 0.5)) - 0.4875$. The Kullback-Leibler (KL) distance between the current and target biased equilibrium distributions  $p_b(x, \tau)$ (\ref {rhoZ}) were calculated using fifteen trajectories for both well-tempered and standard metadynamics. The results are shown as the mean (thick lines) and the standard deviation (pale background). In this case and in all other instances, $p_b(x, \tau)$ is derived directly from the total potential, rather than by accumulating a histogram, as is the usual procedure in applied research. The calculations were performed using the PesMD PLUMED engine \cite{PLMD14} \cite{PLMD19}. For  well-tempered metadynamics $\Delta T = 9$, the initial hill height is 0.05, $\sigma=0.06$. For standard metadynamics gaussians with the same width have a height of 0.001.
    }
    \label{periodic}
\end{figure}

\subsection{Finite interval}
Consider the case where the collective
variable ranges over the finite interval $[0,1]$.
In Sect.~\ref{SectDEFI}, for this case we obtain the differential equation
\[
\lb{defi}
{d\wt V(x,\tau )\/d\tau }={1\/Z_b(\tau)}\int_0^1 \Big(g(x-s)
-\int_0^1 g(u-s)du\Big)e^{-\a\wt V(s,\tau )-\b F(s)}ds,
\]
where $(x,\tau)\in(0,1)\ts\R_+$, $F\in C([0,1])$ is a given positive function,
\[
\lb{defga}
Z_b(\tau)=\int_0^1e^{-\b(F(s)+\wt V(s,\tau))}ds,,
\]
$\a,\b$ are given positive numbers.

Rewrite Eq.~\er{defi} in the form
\[
\lb{eqwtVfipr}
{d\wt V(x,\tau )\/d\tau }={Z_w(\tau)\/Z_b(\tau)}
\Big(\int_0^1g(x-s)p_w(s,\tau)ds-\int_0^1\int_0^1g(x-s)p_w(s,\tau)dsdx\Big),
\]
where the current tempering-reweighted distribution $p_w$ has the form
\[
\lb{defgpZ1}
p_w(s,\tau)={e^{-\a\wt V(s,\tau)-\b F(s)}\/Z_w(\tau)},
\qq
Z_w(\tau)=\int_0^1e^{-\a\wt V(s,\tau)-\b F(s)}ds.
\]

Instead $\tau$, we introduce the new time variable
\[
\lb{thetafi}
\theta(\tau)=\int_0^\tau {Z_w(\tau_1)\/Z_b(\tau_1)}d\tau_1.
\]
We prove in Lemma~\ref{Lmrshtim} that
there exists $\tau_*>0$ such that $\theta:[0,\tau_*)\to [0,+\iy)$.

We denote
$\cV(x,\theta)=\wt V(x,\tau)$.
Then ${d\/d\tau}={Z_w(\tau)\/Z_b(\tau)}{d\/d\theta}$ and Eq.~\er{eqwtVfipr}
takes the form
\[
\lb{eqcVfipr}
{d\cV(x,\theta )\/d\theta }=G\star \gp(x,\theta)-\langle G\star \gp\rangle(\theta),
\]
where
\[
\lb{defgpg*f}
\gp(x,\theta)=p_w(x,\tau)={e^{-\a\cV(s,\theta)-\b F(s)}\/\cZ_w(\theta)},\qq
\cZ_w(\theta)=Z_w(\tau)=\int_0^1e^{-\a\cV(s,\theta)-\b F(s)}ds,
\]
$$
g\star f(x)=\int_0^1g(x-s)f(s)ds,\qq
\langle f\rangle=\int_0^1f(x)dx.
$$

Introduce the class
$$
C_0([0,1])=\{f\in C([0,1]):\int_0^1f(x)dx=0\}
$$
of functions, continuous in $x$ on the interval $[0,1]$
with zero mean value.
The following theorem shows that Eq.~\er{eqcVfipr} has a global
solution in this class.

\begin{theorem}
\lb{Thficl}
Let $F\in C[0,1]$ and let $\cV_0\in C_0([0,1])$.
Then

i) Eq.~\er{eqcVfipr} has a unique solution
$\cV\in C^1(\R_+,C_0([0,1]))$,
satisfying the condition $\cV(x,0)=\cV_0(x),x\in[0,1]$.

ii) The solutions to Eq.~\er{defi} will blow up within a finite time.
\end{theorem}

\no {\bf Remark.} We prove in Proposition~\ref{LmeucV}, that
the solution $\cV$ cannot exhibit complex dynamics.
It means that the solution $\cV(\cdot,\theta)$ to Eq.~\er{eqcVfipr}
 cannot be periodic, in the sense that
$\cV(\cdot,\theta+\gT)\ne\cV(\cdot,\theta)$
for any $\gT>0$. Moreover, the solution $\cV(\cdot,\theta)$ to Eq.~\er{eqcVfipr}
cannot return to its original value in the sense that
$\lim_{n\to\iy}\cV(\cdot,\theta+\theta_n)\ne\cV(\cdot,\theta)$
for any subsequence $(\theta_n)_{n\in\N}$, tending to infinity.

\medskip

Lemma~\ref{Lmstsolfi} shows that Eq.~\er{eqcVfipr},
as well as Eqs.~\er{defi} and \er{eqwtVfipr},
have no stationary solutions in the class $C_0([0,1])$
(and ever in much more larger class of distributions
with support on the interval $[0,1]$).
Therefore, the solution $\cV$ cannot have the limit at $\theta\to\iy$
in this class.
In order to understand this limit, we consider weak solutions.

We prove in Lemma~\ref{Lmglclsol} that
the function $\gp$, given by \er{defgpg*f},
satisfies the equation
\[
\lb{eqgpwt}
{d\gp(x,\theta)\/d\theta}
=-\a\gp(x,\theta)
\big(G\star \gp(x,\theta)-\langle G\star \gp\rangle_\gp(\theta)\big),\qq
(x,\theta)\in(0,1)\ts\R_+,
\]
where
$$
\langle f\rangle_p=\int_0^1f(x)p(x)dx.
$$
Eq.~\er{eqgpwt} belongs to the class of {\it replicator dynamics equations},
see Oechssler and Riedel \cite{OR01}.
Note that the function $F$ only appears
in the initial condition of Eq.~\er{eqgpwt},
$$
\gp(x,0)={e^{-\a\cV(s,0)-\b F(s)}\/\int_0^1e^{-\a\cV(s,0)-\b F(s)}ds}.
$$

Let $\cM[0,1]$ be the set of probability measures on the interval  $[0,1]$.
Each measure $\m\in\cM[0,1]$ is the linear functional on the space $C([0,1])$
of continuous functions on the interval $[0,1]$:
$f\mapsto\int_{[0,1]}f(x)\m(dx)$. Introduce the weak-* topology
on the set $\cM[0,1]$ by
\[
\lb{wcoonv}
\m_n\rightharpoonup\m\qq\Leftrightarrow\qq
\int_{[0,1]}f(x)\m_n(dx)\to\int_{[0,1]}f(x)\m(dx)\ \ \forall\ \
f\in C([0,1]).
\]
For $\m\in\cM[0,1]$ we define the Gaussian potential
\[
\lb{gp}
U_{\m}(x)=\int_{[0,1]}g(x-s)\m(ds)
\]
and its average with respect to the measure
$$
\langle U_\m\rangle_\m
=\int_{[0,1]}U_{\m}(x)\m(dx),
$$
By a weak solution to the equation
\[
\lb{eqgpwtff}
{d\m_\theta\/d\theta}
=-\a\big(U_{\m_\theta}(x)-\langle U_{\m_\theta}\rangle_{\m_\theta}\big)\m_\theta,
\qq \a>0,
\]
we mean a family of probability measures $\m_\theta\in\cM[0,1]$
such that for all $\vp\in C([0,1])$ the function
$\theta\mapsto\int_{[0,1]}\vp(x)\m_\theta(dx)$ is absolutely continuous
and the equation
\[
\lb{eqgpwtf}
{d\/d\theta}\int_{[0,1]}\vp(x)\m_\theta(dx)
=-\a\int_{[0,1]}\vp(x)
\big(U_{\m_\theta}(x)-\langle U_{\m_\theta}\rangle_{\m_\theta}\big)\m_\theta(dx)
\]
holds a.e. on $\theta\in\R_+$.
Weak stationarity of a measure means that
\[
\lb{steqfif}
\int_{[0,1]}\vp(x)
\big(U_{\m}(x)-\langle U_{\m}\rangle_{\m}\big)\m(dx)=0
\]
for all $\vp\in C([0,1])$. This is equivalent to
$
U_{\m}(x)=\langle U_{\m}\rangle_{\m}
$
$\m$-a.e.
Thus, a measure is weakly stationary if and only if its potential
is constant on the support of the measure.

Lemma~\ref{Lmfsmd} shows that
any weak stationary measure from $\cM[0,1]$
is a purely atomic measure with no accumulation points.
We prove the following results.

\begin{theorem}
\lb{Thfiws}
Let $\m_0\in\cM[0,1]$. Then there exists the global weak solution
$(\m_\theta)_{\theta\ge 0}\ss\cM[0,1]$
to Eq.~\er{eqgpwt},
satisfying the initial condition $\m_\theta|_{\theta=0}=\m_0$.
Let $\o(\m_0)$ be the set of the limit points of the trajectory
$\m_{\theta}$,
starting at the point $\m_0$, $\o(\m_0)$ consists of points $\m_\iy$ of the set
$\cM[0,1]$ such that $\m_{t_n}\rightharpoonup\m_\iy$
along some sequence $t_n\to+\iy$. Then

a) Each point $\m_\iy$ of the set $\o(\m_0)$
is a week stationary measure.

b) The set $\o(\m_0)$ consists of finitely atomic measures.

\end{theorem}

\no {\bf Remark.}
1) Proposition \ref{Prop1qs} demonstrates the so-called ``edge effects''.
If $\sigma$ is  small,
then the state $\gp=1$ is a quasistationary solution to Eq.~\er{eqgpwt}
in the following sense:
the change in the distribution within a $\sigma$-neighborhood
of the edges occurs at a rate of order a constant, whereas at a distance
greater than $\sigma$ from the edges,
the change becomes slow--of order $\sigma$.
The solution exits the quasistationary state rapidly, see~Fig.\ref{finite}.

2) In the limit $\s\to 0$ the solution $\gp=1$ is a stationary solution
to Eq.~\er{eqgpwt}.
Moreover, in this case $\o(\m_0)$ consists of the unique element $\m_\iy=1$
for all $\m_0\in\cM[0,1]$.

3) Additional analysis shows that the number of atoms in the limit distribution $\m_\iy$
does not depend on the initial distribution $\m_0$ and is approximately equal to $\s^{-1}$.
In particular, if $\s$ is not too large, then there are peaks not only at the points $0$ and $1$,
but also near the inner points $\s^{-1},2\s^{-1},$ etc, see~Fig.~\ref{finite}~d).

4) Some additional results about the solution $\m_\theta$ are proved in
Theorem~\ref{Thfs}.

\medskip

The following theorem shows that if the solution to Eq.~\er{eqgpwt}
satisfy $\gp(\cdot,0)\in L^1(0,1)$, then $\gp(\cdot,\theta)\in L^1(0,1)$
for all $\theta>0$. Therefore, if the initial measure is absolutely
continuous with respect to the Lebesgue measure, then the atomic measures
cannot arise in a finite time $\theta$. Moreover, the theorem implies that
the set of zeros of the density $\gp$ is the same for all finite $\theta\ge 0$.

\begin{theorem}
\lb{ThL1}
Assume that the initial measure is absolutely continuous with respect to
the Lebesgue measure
$\mu_0(dx)=\gp_0(x)dx$,
where
$\gp_0\in L^1(0,1)$, $\gp_0\ge0$, $\int_0^1\gp_0(x)dx=1$.
Then for each $\theta\ge0$ the measure $\mu_\theta$, satisfying Eq.~\er{eqgpwtff}
and the initial condition $\mu_\theta|_{\theta=0}=\mu_0$,
is absolutely continuous
$\mu_\theta(dx)=\gp(x,\theta)dx$.
Moreover,
$\gp\in C^1\bigl([0,\infty);L^1(0,1)\bigr)$,
$\gp$ satisfies Eq.~\er{eqgpwt} in $L^1(0,1)$ for all $\theta\ge0$
and the initial condition
$\gp(\cdot,0)=\gp_0 $ in $L^1(0,1)$.
Furthermore,
\[
\lb{2estp}
e^{-\sqrt{2\/\pi} {\a\/\s}\theta}\gp_0(x)
\le \gp(x,\theta)\le
e^{\sqrt{2\/\pi} {\a\/\s}\theta}\gp_0(x),
\]
in particular,
\[
\lb{p-p_0=0}
\gp(x,\theta)=0
\quad\Longleftrightarrow\quad
\gp_0(x)=0,
\]
for a.e. $x\in[0,1]$ and for all finite $\theta$.
\end{theorem}

\begin{figure}[htbp]
    \centering
    \input{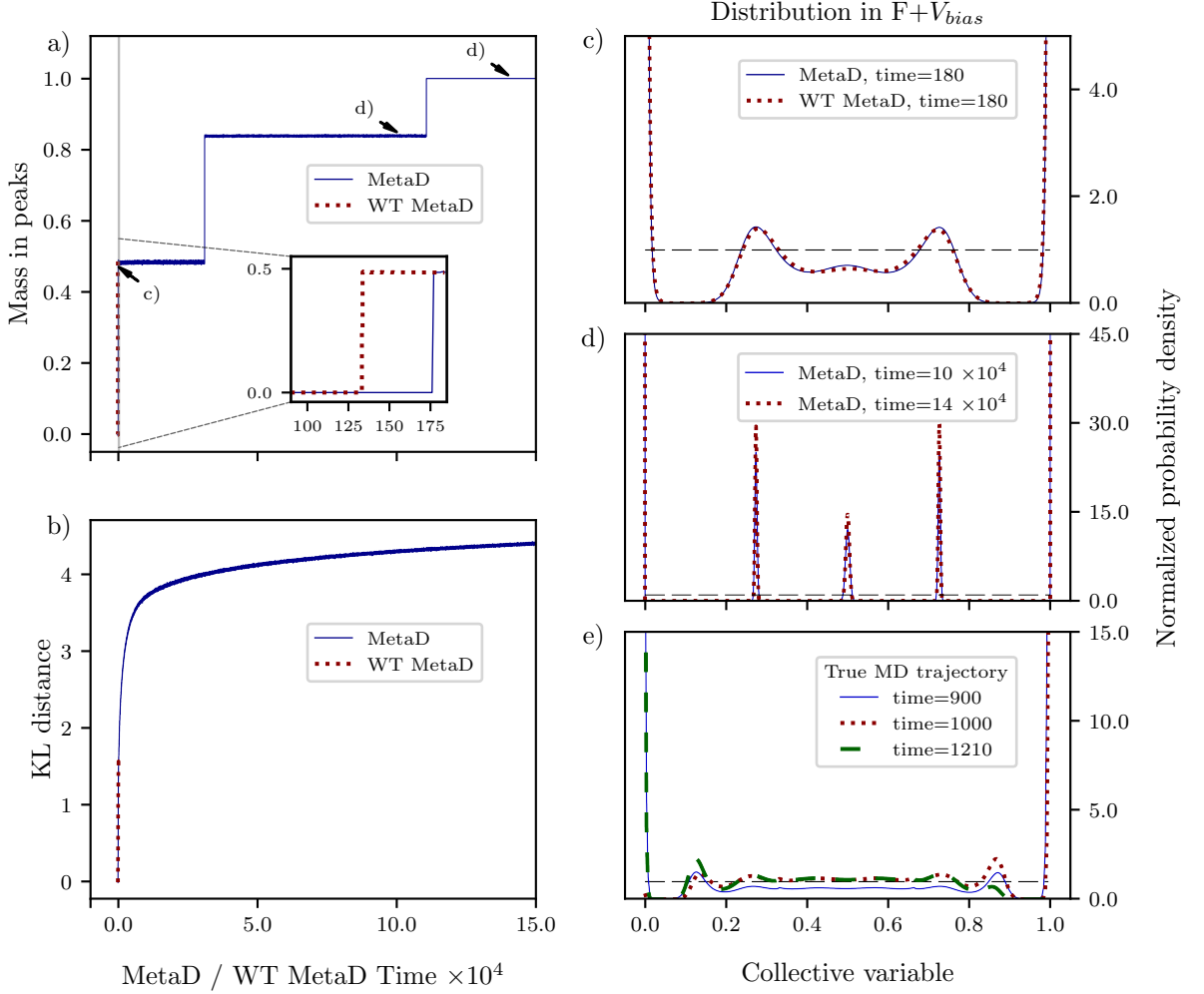}
    \caption{The results of the numerical simulations on the finite interval [0, 1]. a-d) Gaussian functions with $\sigma=0.15$ were deposited within the framework of adiabatic metadynamics in the instantaneous equilibrium limit (see Section 6.1 and \cite{R26}).For well-tempered metadynamics, $\Delta T = 9$, the initial hill height is 1. For standard metadynamics, the hill height is 0.001. Simulations with both metadynamics lasted  equal physical time. The time shown in figures is internal metadynamics time. The original F(x) is is independent of x across the entire interval. Consequently, the target distribution is achieved at the start of the simulation.  However, the system quickly shifts away from a uniform distribution and toward a distribution with very narrow peaks. Peak mass was calculated as the fraction of the total distribution allocated to peaks defined by standard deviation less than 0.005. The applicability of Eq.~\er{dpv} to standard metadynamics is discussed in Section 6.1. e) the results are obtained using a combination of Langevin dynamics and standard metadynamics. The PesMD PLUMED engine is used \cite{PLMD14} \cite{PLMD19}. $x=(3+cos(\wt x)+cos(\wt y)+cos(\wt z))/6$, where $\wt x$, $\wt y$ and $\wt z$ are Cartezian coordinates. $F(x) = 250((x-0.5)^2-0.13)^2$. Gaussian functions with a height of 0.0001 and a $\sigma=0.07$ of 0.07 were deposited every 1,000 steps. Narrowing the peaks renders the system non-adiabatic, which makes Eq.~\er{dpv} inapplicable. In fact, the peaks appear and disappear. This behavior is often observed in real simulations. }
    \label{finite}
\end{figure}

\subsection{Scheme ``INTERVAL''}
We consider the scheme from \cite{BL12}, which is as follows.
The collective variable $x$ ranges over the whole axis.
Gaussians are added at all points of the collective variable visited
by the trajectory; however, outside the interval $[0, 1]$,
the trajectory is determined solely by the initial potential $F$,
and the added potential is not taken into account. Furthermore,
it is assumed that the initial potential grows sufficiently
rapidly as $x \to \pm\infty$. This ensures that the trajectory
does not escape to infinity, while the added potential $V_n(x)$,
$x \in \R$,
decays sufficiently rapidly as $x \to \pm\infty$.

We show in Section~\ref{SectDESI} that in this case
the differential equation has the form
\[
\lb{defiii}
{d\wt V(x,\tau )\/d\tau }={1\/Z_b(\tau)}
\int_\R \Big(g(x-s)-\int_0^1g(u-s)du\Big)
e^{-{\wt V(s,\tau )\/\D T}-{\wt V(s,\tau )+F(s)\/T}}ds,
\qq x\in [0,1],
\]
for all $\tau>0$, and
\[
\lb{wtVext}
\wt V(x,\tau)=\ca \wt V(0,\tau),& x<0\\
\wt V(1,\tau),&x>1\ac.
\]
If we assume
$\int_\R e^{-{F(x)\/T}}dx<\iy$, then the integrals converge.

Assume that
$F\in\cQ$, where $\cQ$ is given by
\[
\lb{defcQF}
\cQ=\{f\in C(\R):f>0,\int_\R e^{-{f(s)\/T}}ds<\iy\}.
\]
Introduce the class $\cX$ of continuous on $\R$ functions,
constant outside of the interval
$[0,1]$ and having zero mean value on this interval
\[
\lb{defcX}
\cX=\{f\in C(\R):\int_0^1f(x)dx=0;f(x)=f(0),\text{ as }x<0;
f(x)=f(1),\text{ as }x>1\}.
\]
Note that $\cX$ is an invariant set on which the evolution of the solutions
to Eq.~\er{defiii} is defined.

Rewrite Eq.~\er{defiii} in the form
\[
\lb{eqwtVfiipr}
{d\wt V(x,\tau )\/d\tau }={Z_w(\tau)\/Z_b(\tau)}
\Big(\int_\R g(x-s)p_w(s,\tau)ds-\int_\R\int_0^1g(u-s)dup_w(s,\tau) ds\Big),
\qq x\in[0,1],
\]
where
\[
\lb{defgpZ1i}
p_w(s,\tau)={e^{-\a\wt V(s,\tau)-\b F(s)}\/Z_w(\tau)},\qq
(s,\tau)\in\R\ts\R_+,
\]
$$
Z_w(\tau)=\int_\R e^{-\a\wt V(s,\tau)-\b F(s)}ds,\qq
 Z_b(\tau)=\int_\R e^{-\b(F(s)+\wt V(s,\tau))}ds,
$$
$\a,\b$ are given by \er{defga}.

Introduce the new variable in place of $\tau$:
$$
\theta(\tau)=\int_0^\tau {Z_w(\tau_1)\/Z_b(\tau_1)}d\tau_1.
$$
Denote
$\cV(x,\theta)=\wt V(x,\tau)$.
Then Eq.~\er{eqwtVfiipr}
takes the form
\[
\lb{eqcVfiipr}
{d\cV(x,\theta )\/d\theta }=g\star \gp(x,\theta)-\langle g\star \gp\rangle(\theta),
\qq (x,\theta)\in(0,1)\ts\R_+,
\]
where
\[
\lb{Aisconvi}
g\star \gp(x,\theta)=\int_\R g(x-s)\gp(s,\theta)ds,\qq
\langle g\star \gp\rangle(\theta)
=\int_\R\int_0^1g(u-s)du\gp(s,\theta)ds,
\]
\[
\lb{defgpg*fi}
\gp(x,\theta)=p_w(x,\tau),\qq
g\star f(x)=\int_\R g(x-s)f(s)ds,\qq
\langle f\rangle=\int_0^1f(x)dx.
\]

\begin{theorem}
\lb{TheucVi}
Let $F\in\cQ$ and let $\cV_0\in\cX$,
where $\cQ$ and $\cX$ are given by \er{defcQF} and \er{defcX},
respectively. Then there exists the unique solution
$\cV\in C^1(\R_+;\cX)$
to Eq.~\er{eqcVfiipr}, satisfying the condition $\cV|_{\theta=0}=\cV_0$.
\end{theorem}

\no {\bf Remark.}
The results, similar to the results of Theorems~\ref{Thficl} and \ref{Thfiws},
hold true,
see Lemma~\ref{LmeucVi} and Theorem~\ref{ThsIws}.

\medskip

Assume, in addition,
that $F(x)=F(0)$ for all $x\in[-L,0]$ and $F(x)=F(1)$ for all $x\in[1,L+1]$
for some $L>0$ large enough.
Let $\cV_L$ denote the biased potential of the form
\[
\lb{defcVL}
\cV_L(x)=-{\b\/\a}\left(-\ol F +\ca
F(0), x<-L,\\
F(x),-L\le x\le L+1,\\
F(1), x>L+1\ac\right),\qq
x\in\R,\qq \ol F=\int_0^1F(x)dx.
\]
In the set $\cX$, consider a ball centered at $\cV_L$:
\[
\lb{bincX}
\cB_r=\{\cV\in\cX:\max_{x\in[0,1]}|\cV(x)-\cV_L(x)|<r\},\qq r>0.
\]
The following theorem shows that if $L$ is large, then
the solution to Eq.~\er{eqcVfiipr},
satisfying the condition
$\cV(\cdot,0)=\cV_L$, remains in its neighborhood for long time.

\begin{theorem}
\lb{ThQSS}
Let $\cV\in\cX$ be a solution to Eq.~\er{eqcVfiipr}
satisfying the condition $\cV(\cdot,0)=\cV_L$.
Let
\[
\lb{defkRLc}
K_R=2e^{2(\b\max_{[0,1]}|F|+\a R)},
\]
 for some $R>0$ and let
\[
\lb{defveL}
\ve_L={Le^{-{L^2\/2\s^2}}\/\sqrt{2\pi}\s^3Z_L}\Big({2\s^2\/L}
+\int_{-\iy}^{-L}e^{\b(F(0)-F(s))}ds
+\int_{L+1}^\iy e^{\b(F(1)-F(s))}ds\Big),
\]
\[
\lb{ZL}
Z_L=1+2L+\int_{-\iy}^{-L}e^{\b(F(0)-F(x))}dx
+\int_{L+1}^\iy e^{\b(F(1)-F(x))}dx.
\]
Then for
\[
\lb{tstgtVinBR}
\theta\le{1\/K_R}\log\Big({K_R\/\ve_L}R+1\Big)
\]
the solution  $\cV$ remains in the ball $\cB_R$.

\end{theorem}

{\no \bf Remark.}
1) $\ve_L\to 0$ at $L\to+\iy$ for fixed $\s$ and at $\s\to 0$
for fixed $L$.

2) The estimate \er{tstgtVinBR} shows that the lifetime
of the quasistationary state is at least
${1\/K_R}\log({K_R\/\ve_L}R+1)$.

3) In Theorem~\ref{ThQSSg} we estimate the difference $|\cV-\cV_L|$.

4) If L is sufficiently large, the system remains
in the vicinity of the quasistationary
state for a long time, see~Fig.~\ref{interval}.

\begin{figure}[htbp]
    \centering
    \input{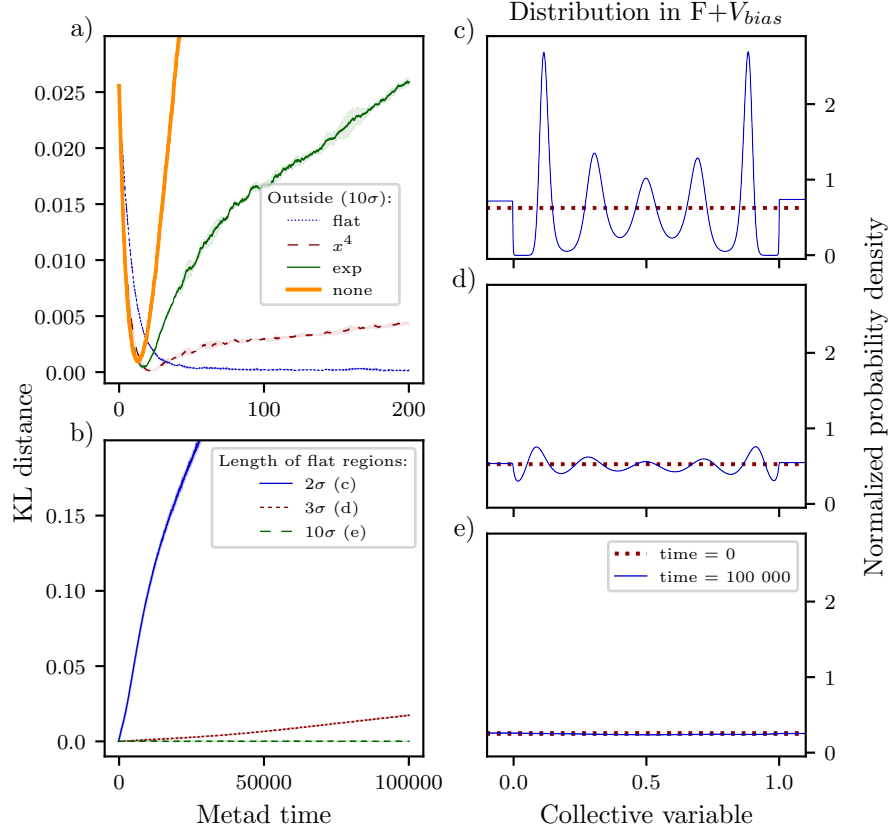}
    \caption{Application of adiabatic metadynamics in the instantaneous-equilibrium limit 
    	(Section 6.1, and \cite{R26}) to the "INTERVAL" scheme. Results are obtained using standard metadynamics with Gaussian functions. Free energy is restored on the interval [0, 1]. The lengths of the regions $x > 0$ and $x < 0$ are chosen to be multiples of Gaussian width ($\sigma$).
    	a) The Gaussian parameters are $h=0.001$ and $\sigma=0.05$. The original free energy function is defined on the interval [0, 1] as $F(x)=50((x-0.5)^2-0.1)^2$. Depending on the shape of the free energy function in the outer regions, the KL distance between the current and target biased equilibrium distributions may reach zero and remain unchanged for a period of time (in the case of flat regions). Alternatively, it may never reach zero (in the case of exponential regions or the absence of regions, which is equivalent to a finite interval). b) The Gaussian parameters are $h=0.001$ and $\sigma=0.15$.  The original F(x) function is independent of x across the entire interval, including the regions outside the interval. Consequently, the target distribution is achieved at the start of the simulation. However, for short outside regions ($2\sigma$ and $3\sigma$), the system quickly shifts away from a uniform distribution toward one with narrow peaks (panels c and d). When the outside regions are long ($10\sigma$), the system preserves the uniform distribution for a longer period of time, and this quasi-stationary state can, in principle, be captured (panel e). }
    \label{interval}
\end{figure}

\section{Periodic case}
\setcounter{equation}{0}

\subsection{Stationary solution}
We consider Eq.~\er{eqcVfiprpc}.
The corresponding stationary equation has the form
\[
\lb{depcst}
\int_0^1\big(G(x-s)-1\big)\gp(s)ds=0,\qq x\in\T.
\]

For the function $f\in C(\T)$
we introduce the Fourier coefficients $\wh f_m,m\in\Z$:
\[
\lb{FC}
f(x)=\sum_{m\in\Z}e^{i2\pi mx}\wh f_m,\qq
\wh f_m=\int_{-\iy}^\iy f(x)e^{-i2\pi mx}dx.
\]
In particular,
\[
\lb{FCG}
G(x)=\sum_{m\in\Z}e^{i2\pi mx}\wh G_m,\qq \wh G_0=1,\qq
\wh G_m={1\/\s\sqrt{2\pi}}\int_{-\iy}^\iy e^{-{x^2\/2\s^2}}e^{-i2\pi mx}dx
=e^{-2\s^2\pi^2m^2}>0,
\]
for all $m\in\Z$.

Introduce the class of densities
\[
\lb{cP01}
\cP[0,1]=\{p\in C[0,1]:p>0,\int_0^1p(x)dx=1\}.
\]

\begin{lemma} Let $F\in C(\T)$. Then
 Eq.~\er{eqcVfiprpc} has the unique solution $\gp=1$
 in the class $\cP[0,1]$, given by \er{cP01},
 and
Eq.~\er{depcst} has the unique solution
\[
\lb{solsteq}
\cV(x)={\ol F-F(x)\/{T\/\D T}+1},\qq \ol F=\int_0^1F(x)dx.
\]
in the class $C_0(\T)$.
\end{lemma}

\no {\bf Proof.}
Rewrite Eq.~\er{eqcVfiprpc} in the form
\[
\lb{steq}
\int_0^1 \gp(s)ds=\int_0^1G(x-s)\gp(s)ds,\qq x\in\T.
\]
Search for the solution in the class $\gp\in \cP[0,1]$, given by \er{cP01}.
Due to $\int_0^1G(x-s)ds=1$, the function $\gp(x)=\const$
satisfies Eq.~\er{steq}. Conversely,
let $\gp\in \cP[0,1]$ is a solution to Eq.~\er{steq}.
Eq.~\er{steq} takes the form
$$
\wh \gp_0=\sum_{m\in\Z}e^{i2\pi mx}\wh G_m\wh \gp_m.
$$
This gives $\wh \gp_m=0,m\ne 0$, therefore, $\gp(x)=\wh \gp_0$.
This, $y(x)$ is a solution to Eq.~\er{steq}
iff $\gp(x)=\const$.
The definition \er{defgpg*f} implies that
Eq.~\er{depcst} has the unique solution \er{solsteq} in the class
$C_0(\T)$.~\BBox

\subsection{Local solvability}

Introduce the norm in the space $ C(\T)$ by $\|f\|=\max_{x\in\T}|f(x)|$.

\begin{lemma}
\lb{Lmuthvpc}
Let $\cV_0\in C_0(\T)$,
let $M>0$, and let $\cT\in(0,\min\{{M\/2 A_1},{1\/2 A_2}\})$, where
\[
\lb{defmA}
 A_1={B_2\/B_1},\qq
 A_2={\a B_2\/ B_1}(  B_2+1),
\]
\[
\lb{defmA1}
 B_1=\int_0^1e^{-\b F(s)}ds,\qq
 B_2=e^{2\a(\|\cV_0\|+M)}.
\]
Then Eq.~\er{eqcVfiprpc} has the unique solution
$\cV\in C^1([0,\cT],C_0(\T))$
satisfying the conditions:

1) $\max_{\tau\in[0,\cT]}\|\cV(\cdot,\tau)-\cV_0\|\le M$,

2) $\cV(x,0)=\cV_0(x),x\in\T$.
\end{lemma}

\no {\bf Proof.} Rewrite Eq.~\er{eqcVfiprpc} in the form
\[
\lb{ideq2}
{d\cV(x,\theta )\/d\theta }=\int_0^1K(x-s)\gp(s,\theta;\cV)ds,
\]
where
\[
\lb{defKrho}
K(x)=G(x)-1,\qq
\gp(s,\theta;\cV)
=\frac{Y(s,\theta,\cV)}{\cZ_w(\theta,\cV)},
\]
\[
\lb{defYZ}
Y(s,\theta,\cV)=e^{-\a\cV(s,\theta )-\b F(s)},
\qq \cZ_w(\theta,\cV)=\int_0^1e^{-\a\cV(s,\theta )-\b F(s)}ds.
\]
Rewrite this equation in the form of the integral equation
\[
\lb{ieqvpc}
\cV(x,\theta )=(\Phi \cV)(x,\theta),\qq
(x,\theta)\in\T\ts[0,\cT],
\]
where the operator $\Phi$, given by the identity
$$
(\Phi \cV)(x,\theta)
=\cV_0(x)+\int_0^\theta\int_0^1K(x-s)\gp(s,t;\cV)dsdt,
$$
acts in the normed space
$$
\cM=\{\cV(x,\theta):
\cV(\cdot,\theta)\in C_0(\T)\text{ for all }\theta\in[0,\cT],
\cV(x,\cdot)\in C[0,\cT]\text{ for all } x\in\T\},
$$
with the norm
$$
\|\cV\|_{\cM}=\max_{\theta\in[0,\cT]}\|\cV(\cdot,\theta)\|.
$$
For all functions $\cV_1$ and $\cV_2$ from the space $\cM$
and for all $(x,\theta)\in\T\ts[0,\cT]$ we have
$$
(\Phi \cV_1)(x,\theta)-(\Phi \cV_2)(x,\theta)
=\int_0^\theta\int_0^1 K(x-s)
\big(\gp(s,t;\cV_1)-\gp(s,t;\cV_2)\big)dsdt.
$$
This identity and the estimate $|\int_0^1 K(x-s)ds|\le 2$ give
\[
\lb{estwtVn-V0}
\|(\Phi \cV)(\cdot,\theta )-\cV_0(\cdot,\theta )\|
=2\int_0^\theta\max_{s\in\T}
\big|\gp(s,t;\cV)\big|dt,
\]
\[
\lb{estwtVn-Vn-1}
\|(\Phi \cV_1)(\cdot,\theta )-(\Phi \cV_2)(\cdot,\theta )\|
=2\int_0^\theta\max_{s\in\T}
\big|\gp(s,t;\cV_1)-\gp(s,t;\cV_2)\big|dt,
\]
for all $\theta\in[0,\cT]$.
Consider the closed set
$$
\cB=\{\cV\in \cM:
\|\cV-\cV_0\|_{\cM}\le M\}
$$
in the space $\cM$.
The definition \er{defYZ} gives
\[
\lb{estZtau}
{ B_1\/ B_2^{1\/2}}
\le \cZ_w(\theta,\cV)
\le
 B_1 B_2^{1\/2},\qq
Y(s,\theta;\cV)\le  B_2^{1\/2},
\]
for all $(s,\theta)\in\T\ts[0,\cT]$ and for all $\cV\in\cB$,
here $ B_1, B_2$ are given by \er{defmA1},
$$
\begin{aligned}
|Y(s,\theta;\cV_1)-Y(s,\theta;\cV_2)|\le
\Big|e^{-\b F(s)}\big(e^{-\a\cV_1(s,\theta)}
-e^{-\a\cV_2(s,\theta)}\big)\Big|
\\
\le \a B_2^{1\/2}|\cV_1(s,\theta )-\cV_2(s,\theta )|,
\end{aligned}
$$
for all $\cV_1,\cV_2\in\cB$,
here we used the estimate $|e^{\a}-e^{\b}|\le e^{\max\{|\a|,|\b|\}}|\a-\b|$.
The estimate \er{estZtau} gives
$$
\begin{aligned}
\Big|{1\/\cZ_w(\theta,\cV_1)}-{1\/\cZ_w(\theta,\cV_2)}\Big|
=\Big|{\cZ_w(\theta,\cV_2)-\cZ_w(\theta,\cV_1)\/\cZ_w(\theta,\cV_2)\cZ_w(\theta,\cV_1)}\Big|
\le { B_2\/ B_1^2}\Big|\int_0^1e^{-\b F(s)}\big(e^{-\a\cV_1(s,\theta)}
-e^{-\a\cV_2(s,\theta)}\big)ds\Big|
\\
\le { B_2\/ B_1}\max_{s\in[0,1]}\big|e^{-\a\cV_1(s,\theta)}
-e^{-\a\cV_2(s,\theta)}\big|
\le{\a B_2^{3\/2}\/ B_1}
\big\|\cV_1(\cdot,\theta)-\cV_2(\cdot,\theta)\big\|.
\end{aligned}
$$
This estimate and \er{defKrho} give
$$
\gp(s,\theta;\cV)
\le\frac{|Y(s,\theta,\cV)|}{|\cZ_w(\theta,\cV)|}\le A_1,
$$
$$
\begin{aligned}
\big|\gp(s,\theta;\cV_1)-\gp(s,\theta;\cV_2)\big|
\le
Y(s,\theta;\cV_1)\Big|{1\/\cZ_w(\theta,\cV_1)}-{1\/\cZ_w(\theta,\cV_2)}\Big|
+{|Y(s,\theta;\cV_1)-Y(s,\theta;\cV_2)|\/\cZ_w(\theta,\cV_2)}
\\
\le  A_2\big\|\cV_1(\cdot,\theta)-\cV_2(\cdot,\theta)\big\|.
\end{aligned}
$$
where $ A_1, A_2$ are given by \er{defmA}.
Substituting these estimates into \er{estwtVn-V0} and
 \er{estwtVn-Vn-1}, we obtain
$$
\|\Phi \cV-\cV_0\|_{\cM}
\le 2\cT A_1.
$$
This estimate shows that if $\cT<{M\/2 A_1}$, then $\Phi$
maps the set $\cB$ into itself. Moreover, we have
$$
\|(\Phi\cV_1)(\cdot,\theta )-(\Phi\cV_2)(\cdot,\theta )\|
\le 2 A_2\int_0^\theta\|\cV_1(\cdot,t )-\cV_2(\cdot,t )\|dt,
$$
which yields
$$
\|\Phi\cV_1-\Phi\cV_2\|_{\cM}
\le 2 A_2 \cT\|\cV_1-\cV_2\|_{\cM},
$$
If $\cT<\min\{{M\/2 A_1},{1\/2 A_2}\}$, then the mapping $\Phi$ is contractive
on the set $\cB$. Therefore, the integral equation
\er{ieqvpc} has the unique solution in $\cB$.
This solution is differentiable with respect to $\theta$.
Thus, it is the solution to Eq.~\er{depc}
on $\T\ts[0,\cT]$, satisfying the condition 1)--2).~\BBox

\subsection{Global solvability}

\begin{lemma}
\lb{Thuthvpc}
The solution  $\cV$, defined in Lemma~\ref{Lmuthvpc},
can be extended from the interval $[0,\cT]$ of the variable $\tau$
onto the whole half-line $\R_+$.
The extended solution satisfies $\cV\in C^1(\R_+,C_0(\T))$ and
$\cV(x,0)=\cV_0(x),x\in\T$.
Moreover,
\[
\cV(x,\tau)\to {\ol F-F(x)\/{T\/\D T}+1}
\]
as $\tau\to\iy$, uniformly on $x\in\T$.

\end{lemma}

\no {\bf Proof.}
The definition \er{defYZ} implies
$$
\frac{d Y(s,\theta)}{d\theta}
=-\a Y(s,\theta)\frac{d\cV(s,\theta)}{d\theta},
$$
where $\a$ is given by \er{alhabeta}.
Substituting this identity into \er{ideq2} and using \er{defKrho},
we obtain the equation for $Y$:
\[
\lb{nseq}
\frac{d Y(x,\theta)}{d\theta}
=\a\eta(\theta) Y(x,\theta)\int_0^1\big(1-G(x-s)\big)Y(s,\theta)ds,\qq
\eta(\theta)=\frac{\a}{\cZ_w(\theta)}.
\]
This equation has the solution $Y\in C^1([0,\cT],C(\T))$
satisfying the conditions:

1) $Y(x,\theta)>0$ for all $(x,\theta)\in\T\ts[0,\cT]$,

2) $Y(x,0)=Y_0(x)$ for all $x\in\T$, where
$Y_0(x)=e^{-\a\cV_0(x)-\b F(x)}$.

Assume that the solution $Y(x,\theta)$ to Eq.~\er{nseq}, continuous on $\T$
and satisfying the condition $Y(x,0)=Y_0(x)>0,x\in\T$, exists
on the interval
$[0,\cT_*)$ for some $\cT_*>\cT$ and
$\|Y(\cdot,\theta)\|\to\iy$, as $\theta\to\cT_*-0$.
Then, rewriting Eq.~\er{nseq} in the form
$$
\frac{d Y(x,\theta)}{d\theta}+\a Y(x,\theta)\int_0^1G(x-s)\gp(s,\theta)ds
=Y(x,\theta),
$$
we obtain
\[
\lb{ieq4pc}
\frac{d}{d\theta}\Big(Y(x,\theta)
e^{\a \int_0^\theta d\x\int_0^1G(x-s)\gp(s,\x)ds}\Big)
=e^{\a \int_0^\theta d\x\int_0^1G(x-s)\gp(s,\x)ds}Y(x,\theta),
\]
and rewriting Eq.~\er{nseq} in the form
$$
\frac{d Y(x,\theta)}{d\theta}-Y(x,\theta)=-\a Y(x,\theta)\int_0^1G(x-s)\gp(s,\theta)ds,
$$
we obtain
\[
\lb{ieq5pc}
\frac{d}{d\theta}\Big(Y(x,\theta)
e^{-\theta }\Big)
=-\a e^{-\theta}Y(x,\theta)
\int_0^1G(x-s)\gp(s,\theta)ds,
\]
for all $x\in\T$. The identity \er{ieq4pc} gives that for all $x\in\T$
the function
\[
\lb{Yvozr}
Y(x,\theta)e^{\a \int_0^\theta d\x\int_0^1G(x-s)\gp(s,\x)ds}
\qq\text{is increasing in $\theta$ on the interval $(0,\cT_*)$.}
\]
The identity \er{ieq5pc} yields that for all $x\in\T$
the function
\[
\lb{Yubyv}
Y(x,\theta)e^{-\theta}\qq
\text{is decreasing in $\theta$ on the interval $(0,\cT_*)$.}
\]
The definitions \er{defYZ} imply
\[
\lb{defZthY}
\cZ_w(\theta)=\int_0^1Y(s,\theta)ds,\qq\theta\in[0,\cT_*).
\]
Substituting this identity into the expression for $\eta$ from \er{nseq},
we obtain
$$
\eta(\theta)=\frac{\a }{\int_0^1Y(s,\theta)ds},
\qq\theta\in[0,\cT_*).
$$
The identity \er{Yvozr} gives
\[
\lb{Yvozr1}
Y(x,\theta)e^{\a \int_0^\theta d\x\int_0^1G(x-s)\gp(s,\x)ds}\ge Y_0(x)>0
\qq\forall\qq(x,\theta)\in\T\ts[0,\cT_*).
\]
Let $x\in\T$. The estimate
\er{Yvozr1} gives
\[
\lb{Yvozr2}
Y(x,\theta)\ge Y_0(x)e^{-\a \int_0^{\theta} d\x\int_0^1G(x-s)\gp(s,\x)ds},
\]
for all $\theta\in[0,\cT^*)$.
Due to our hypothesis, for all $\theta\in[0,\cT_*)$
there exists the solution $Y(x,\theta)$ to Eq.~\er{nseq},
continuous on $\T$. Then the estimate \er{Yvozr2} gives
$Y(x,\theta)>0$ for all $\theta\in[0,\cT^*)$.
The continuity $Y$ with respect to $\theta$ implies $Y(x,\cT_*)\ge 0$.
Moreover, due to our hypothesis,
$\|Y(\cdot,\theta)\|\to\iy$, as $\theta\to\cT_*-0$.
This yields that $Y(x,\cT_*)$ is not an identical zero.
Then the definition \er{defZthY} yields, that $\cZ_w(\theta)>0$ on the whole interval
$[0,\cT_*]$, including the endpoints. Then the definition \er{nseq}
gives that $\eta(\theta)$ is bounded on the whole interval
$[0,\cT_*]$.

The relation \er{Yubyv} implies
\[
\lb{Yubyv1}
Y(x,\theta)e^{-\theta }\le Y_0(x)\qq
\forall\qq (x,\theta)\in\T\ts[0,\cT_*).
\]
Integrating Eq.~\er{nseq}, we obtain
$$
{d\/d\theta }\int_0^1 Y(x,\theta )dx
=\eta(\theta)\bigg(\Big(\int_0^1Y(x,\theta )dx\Big)^2
-\int_0^1\int_0^1G(x-s)Y(x,\theta )Y(s,\theta )dsdx\bigg),
$$
for all $\theta\in(0,\cT_*)$. The identities
$$
\Big(\int_0^1Y(x,\theta )dx\Big)^2=\wh Y_0^2(\theta),\qq
\int_0^1\int_0^1G(x-s)Y(x,\theta )Y(s,\theta )dsdx
=\sum_{m\in\Z}\wh Y_m(\theta)\wh Y_{-m}(\theta)\wh G_m,
$$
give
\[
\lb{Ftrid}
\wh Y_0'(\theta)=-\sum_{m\ne 0}|\wh Y_m(\theta)|^2\wh G_m,
\]
here $\wh Y_m(\theta)$ are the Fourier coefficients of the function $Y(x,\theta)$,
Therefore, $\wh Y_0'(\theta)\le 0$ for all $\t\in(0,\cT_*)$.
Then the function $\wh Y_0(\theta)$ decreases on $(0,\cT_*)$.
This yields
$$
\wh Y_0(\theta)=\int_0^1 Y(x,\theta )dx\le \int_0^1 Y_0(x)dx
$$
for all $\theta\in(0,\cT_*)$.
The estimate \er{Yubyv1} gives
$$
Y(x,\theta)\le Y_0(x)e^{\theta}\qq
\forall\qq (x,\theta)\in\T\ts[0,\cT_*).
$$
Due to the continuity, we have
$$
Y(x,\theta)\le Y_0(x)e^{\cT_*}\qq
\forall\qq (x,\theta)\in\T\ts[0,\cT_*].
$$
It contradicts to our hypothesis
$\|Y(\cdot,\theta)\|\to\iy$, as $\theta\to\cT_*-0$,
therefore $Y$ can be extended to the solution for all $\theta>0$ and
$Y(x,\theta)>0$ for all $(x,\theta)\in\T\ts\R_+$.

The function $\wh Y_0(\theta)$ is decreasing and bounded below,
since $\wh Y_0(\theta)>0$
for all $\t>0$.
This yields $\wh Y_0(\theta)\to\const\ge 0$, as $\theta\to\iy$.
Then $\wh Y_0'(\theta)\to 0$, as $\theta\to\iy$.
The identity \er{Ftrid} implies $\wh Y_m(\theta)\to 0$
for all $m\in\Z\sm 0$. Therefore, $Y(x,\theta)$ tends to a function, constant
on $\T$. The definition \er{defYZ} of the function $Y$ gives all statements
about the solution $\cV$.~\BBox

\medskip

{\bf Proof of Theorem~\ref{Thper}.}
The result follows from Lemma~\ref{Thuthvpc}.~\BBox

\section{Finite interval}
\setcounter{equation}{0}

\subsection{Differential equation}
\lb{SectDEFI}
Consider the case where the collective
variable ranges over the finite interval $[0,1]$.
In this case the integral $\ol V_n=\int_0^1V_n(x)dx$
is the mean value of the potential.
The function $\cG(x,s)$ has the form
$$
\cG(x,s)=hg(x-s),
$$
where $g$ is the Gaussian given by \er{gaussp}
$h>0$ is the initial height of the Gaussian.
Then Eq.~\er{ie} gives the following equations for $\ol V_n$
and for the rest part
$\wt V_n(x)=V_n(x)-\ol V_n$ of the potential:
$$
\ol V_{n+1}=\ol V_n+he^{-{\ol V_n\/\D T}}e^{-{\wt V_n(x_{n+1})\/\D T}}
\int_0^1g(x-x_{n+1})dx,
$$
\[
\lb{ie2}
\wt V_{n+1}(x)=\wt V_n(x)
+he^{-{\ol V_n\/\D T}}e^{-{\wt V_n(x_{n+1})\/\D T}}
\Big(g(x-x_{n+1})-\int_0^1g(x-x_{n+1})dx\Big).
\]
Introduce the new time scale
$\tau(t_n)=h\sum_{j=0}^{n-1}e^{-{\ol V_j\/\D T}}$.
Eq.~\er{ie2} takes the form
\[
\lb{ie3}
{\wt V_{n+1}(x)-\wt V_n(x)\/\tau(t_{n+1})-\tau(t_n)}
=e^{-{\wt V_n(x_{n+1})\/\D T}}
\Big(g(x-x_{n+1})-\int_0^1g(x-x_{n+1})dx\Big).
\]

Introduce the function $V(x,\tau)$
so that $V(x,\tau(t_n))=V_n(x)$ for all $x\in[0,1]$ and for all
$n\in\N$ large enough.
Similarly, we introduce the functions $\wt V(x,\tau)$ and $\ol V(\tau)$.
After a sufficient amount of time,
the number of visits to each state will be distributed according
to the Boltzmann law: the density $ p_b(x,\tau)$
of the system's state distribution at time $\tau$ will be equal to
\[
\lb{rhoZ1}
 p_b(x,\tau)=\Big(\int_0^1e^{-{F(s)+V(s,\tau)\/T}}ds\Big)^{-1}
 e^{-{F(x)+V(x,\tau)\/T}}
={ e^{-{F(x)+\wt V(x,\tau)\/T}}\/Z_b(\tau)},\qq
Z_b(\tau)=\int_0^1e^{-{F(s)+\wt V(s,\tau)\/T}}ds.
\]
Asymptotically, for large $n$,
Eq.~\er{ie3} reduces to the differential equation \er{defi}.

\subsection{Local solvability} We prove the local solvability
of Eq.~\er{eqcVfipr}.

\begin{lemma}
\lb{Thlocsolfi}
Let $\cV_0\in C_0([0,1])$.
Let $M>0$ and let $\cT\in(0,\min\{{M\/2 A_1},{1\/2 A_2}\})$, where
\[
\lb{defmAq}
 A_1={ B_2\/ B_1},\qq
 A_2={\a B_2\/ B_1}( B_2+1),
\]
\[
\lb{defmA1q}
 B_1=\int_0^1e^{-\b F(s)}ds,\qq B_2=e^{2\a(\|\cV_0\|+M)}.
\]
Then Eq.~\er{eqcVfipr} has the unique solution
$\cV\in C^1([0,\cT],C_0([0,1]))$,
satisfying the conditions

1) $\max_{\tau\in[0,\cT]}\|\cV(\cdot,\tau)-\cV_0\|\le M$,

2) $\cV(x,0)=\cV_0(x),x\in[0,1]$.
\end{lemma}

\no {\bf Proof.}
Rewrite Eq.~\er{eqcVfipr} in the form
\[
\lb{ideq2q}
{d\cV(x,\theta )\/d\theta }=\int_0^1K(x,s)\gp(s,\theta)ds,
\]
where
\[
\lb{defKrhoq}
K(x,s)=g(x-s)-\int_0^1g(x-u)du,
\qq
\gp(s,\theta;\cV)
=\frac{Y(s,\theta,\cV)}{\cZ_w(\theta,\cV)},
\]
\[
\lb{defYZq}
Y(s,\theta,\cV)=e^{-\a\cV(s,\theta )-\b F(s)},
\qq \cZ_w(\theta,\cV)=\int_0^1e^{-\a\cV(s,\theta )-\b F(s)}ds.
\]
Rewrite this equation in the form of the integral equation
\[
\lb{ieqvpc1}
\cV(x,\theta )=(\Phi \cV)(x,\theta),\qq
(x,\theta)\in[0,1]\ts[0,\cT],
\]
where the operator $\Phi$, given by
$$
(\Phi \cV)(x,\theta)
=\cV_0(x)+\int_0^\theta\int_0^1K(x-s)\gp(s,\a)dsd\a,
$$
acts in the norm space
$$
\cM=\{\cV(x,\theta):
\cV(\cdot,\theta)\in C_0[0,1]\text{ for all }\theta\in[0,\cT],
\cV(x,\cdot)\in C[0,\cT]\text{ for all } x\in[0,1]\},
$$
with the norm
$$
\|\cV\|_{\cM}=\max_{\theta\in[0,\cT]}\|\cV(\cdot,\theta)\|.
$$
For all functions $\cV_1$ and $\cV_2$ from the space $\cM$
and for all $(x,\theta)\in[0,1]\ts[0,\cT]$ we have
$$
(\Phi \cV_1)(x,\theta)-(\Phi \cV_2)(x,\theta)
=\int_0^\theta\int_0^1 K(x-s)
\big(\gp(s,\a;\cV_1)-\gp(s,\a;\cV_2)\big)dsd\a.
$$
This identity and the estimate $|\int_0^1 K(x-s)ds|\le 2$ give
\[
\lb{estwtVn-V0q}
\|(\Phi \cV)(\cdot,\theta )-\cV_0(\cdot,\theta )\|
=2\int_0^\theta\max_{s\in[0,1]}
\big|\gp(s,\a;\cV)\big|d\a,
\]
\[
\lb{estwtVn-Vn-1q}
\|(\Phi \cV_1)(\cdot,\theta )-(\Phi \cV_2)(\cdot,\theta )\|
=2\int_0^\theta\max_{s\in[0,1]}
\big|\gp(s,\a;\cV_1)-\gp(s,\a;\cV_2)\big|d\a,
\]
for all $\theta\in[0,\cT]$.
Consider the closed set
$$
\cB=\{\cV\in \cM:
\|\cV-\cV_0\|_{\cM}\le M\}
$$
in the space $\cM$.
For all $(s,\theta)\in[0,1]\ts[0,\cT]$ we have from \er{defYZq}
\[
\lb{estZtauq}
{ B_1\/ B_2^{1\/2}}
\le \cZ_w(\theta,\cV)
\le
 B_1 B_2^{1\/2},\qq
Y(s,\theta;\cV)\le h B_2^{1\/2},
\]
for all $\cV\in\cB$,
$ B_1, B_2$ are given by \er{defmA1q},
$$
|Y(s,\theta;\cV_1)-Y(s,\theta;\cV_2)|\le
\Big|e^{-\b F(s)}\big(e^{-\a\cV_1(s,\theta)}
-e^{-\a\cV_2(s,\theta)}\big)\Big|
\le h\a B_2^{1\/2}|\cV_1(s,\theta )-\cV_2(s,\theta )|,
$$
for all $\cV_1,\cV_2\in\cB$,
here we used the estimate $|e^{x}-e^{y}|\le e^{\max\{|x|,|y|\}}|x-y|$.
The estimate \er{estZtauq} gives
$$
\begin{aligned}
\Big|{1\/\cZ_w(\theta,\cV_1)}-{1\/\cZ_w(\theta,\cV_2)}\Big|
=\Big|{\cZ_w(\theta,\cV_2)-\cZ_w(\theta,\cV_1)\/\cZ_w(\theta,\cV_2)\cZ_w(\theta,\cV_1)}\Big|
\le { B_2\/ B_1^2}\Big|\int_0^1e^{-\b F(s)}\big(e^{-\a\cV_1(s,\theta)}
-e^{-\a\cV_2(s,\theta)}\big)ds\Big|
\\
\le { B_2\/ B_1}\max_{s\in[0,1]}\big|e^{-\a\cV_1(s,\theta)}
-e^{-\a\cV_2(s,\theta)}\big|
\le{\a B_2^{3\/2}\/ B_1}
\big\|\cV_1(\cdot,\theta)-\cV_2(\cdot,\theta)\big\|.
\end{aligned}
$$
Then we obtain from \er{defKrhoq}
$$
\gp(s,\theta;\cV)
\le\frac{|Y(s,\theta,\cV)|}{|\cZ_w(\theta,\cV)|}\le A_1,
$$
$$
\begin{aligned}
\big|\gp(s,\theta;\cV_1)-\gp(s,\theta;\cV_2)\big|
\le
Y(s,\theta;\cV_1)\Big|{1\/Z_b(\theta,\cV_1)}-{1\/Z_b(\theta,\cV_2)}\Big|
+{|Y(s,\theta;\cV_1)-Y(s,\theta;\cV_2)|\/Z_b(\theta,\cV_2)}
\\
\le  A_2\big\|\cV_1(\cdot,\theta)-\cV_2(\cdot,\theta)\big\|.
\end{aligned}
$$
where $ A_1, A_2$ are given by \er{defmAq}.
Substituting these estimates into \er{estwtVn-V0q} and
 \er{estwtVn-Vn-1q}, we obtain
$$
\|\Phi \cV-\cV_0\|_{\cM}
\le 2\cT A_1.
$$
This estimate show that if $\cT<{M\/2 A_1}$, then $\Phi$
maps the set $\cB$ into itself. Moreover, we have
$$
\|(\Phi\cV_1)(\cdot,\theta )-(\Phi\cV_2)(\cdot,\theta )\|
\le 2 A_2\int_0^\theta\|\cV_1(\cdot,\a )-\cV_2(\cdot,\a )\|d\a,
$$
which yields
$$
\|\Phi\cV_1-\Phi\cV_2\|_{\cM}
\le 2 A_2 \cT\|\cV_1-\cV_2\|_{\cM},
$$
If $\cT<\min\{{M\/2 A_1},{1\/2 A_2}\}$, then the mapping $\Phi$ is
contractive on the set $\cB$. Therefore, the integral equation
\er{ieqvpc1} has a unique solution in $\cB$.
This solution is differentiable with respect to $\theta$ function.
Therefore, this is a solution to Eq.~\er{eqcVfipr}
on $[0,1]\ts[0,\cT]$, satisfying the conditions 1)--2).~\BBox

\subsection{Global solvability}
We will show that Eq.~\er{eqcVfipr} has a global solution.

\begin{lemma}
\lb{LmeucV}
The solution  $\cV$, defined in Lemma~\ref{Thlocsolfi},
has a unique continuation $\cV\in C^1(\R_+;C_0[0,1])$.
\end{lemma}

\no {\bf Proof.} Let $\cV(x,\theta)$ be the solution to Eq.~\er{eqcVfipr},
satisfying the condition $\cV(x,0)=\cV_0(x),x\in[0,1]$.
Assume that there exists $T_{max}<\iy$ such that
the solution $\cV(\cdot,\theta )$ on the interval
$\theta\in[0,T_{max}]$, and this solution cannot be extended
onto the larger interval $\theta\in[0,T_{max}+\d]$, for any $\d>0$.

Eq.~\er{eqcVfipr} gives
$$
\max_{x\in[0,1]}\Big|{d\cV(x,\theta )\/d\theta }\Big|
\le 2\max_{x\in[0,1]}|g(x)|\int_0^1\gp(x,\theta)dx=2\max_{x\in[0,1]}|g(x)|,\qq
\forall\ \ \theta\ge 0.
$$
This yields
\[
\lb{estcVt1-cVt2}
\max_{x\in[0,1]}\big|\cV(x,\theta_1)-\cV(x,\theta_2)\big|\le
2\max_{x\in[0,1]}|g(x)||\theta_1-\theta_2|,\qq \forall\ \ \theta_1,\theta_2\ge 0.
\]
Since the space $C_0([0,1])$ is complete, the estimate \er{estcVt1-cVt2}
yields that the limit $\cV_*=\lim_{\theta\to T_{max}}\in C_0([0,1])$
exists. Lemma~\ref{Thlocsolfi}
implies that there exists the unique solution $\cV(\cdot,\theta)$,
$\theta\in[0,T_{max}+\d]$, for some $\d>0$, satisfying the condition
$\cV(\cdot,T_{max})=\cV_*$.
Therefore, the solution $\cV(x,\theta)$ can be extended
uniquely from $[0,T_{max}]$
 onto a larger interval, which gives the result.~\BBox

\medskip

\no {\bf Proof of Theorem~\ref{Thficl}~i).} The result follows from
Lemma~\ref{LmeucV}.~\BBox

\subsection{Stationary equation}
Theorem~\ref{Thficl} shows that
Eq.~\er{eqcVfipr} has a global solution. We will now consider
the dynamics of this solution at large time $\theta$. First, we will show that
a stationary solution does not exist.
Consider the stationary equation
\[
\lb{steqfi1}
0=\int_0^1 \Big(g(x-s)
-\int_0^1 g(u-s)du\Big)\gp(s)ds,
\qq x\in [0,1].
\]
Its solution corresponds to the equilibrium state.
Rewrite Eq.~\er{steqfi1} in the form
\[
\lb{steqfi}
\int_0^1g(x-s)\gp(s)ds
=C,\qq x\in[0,1],\qq C=\int_0^1\int_0^1g(u-s)\gp(s)duds.
\]

\begin{lemma}
\lb{Lmstsolfi}
Eq.~\er{steqfi} has no non-trivial solutions in
the class of distributions with support on the interval $[0,1]$.

\end{lemma}

\no {\bf Proof.}
Let the solution $\gp$ to Eq.~\er{steqfi} be a distribution
 with support on the interval $[0,1]$.
Due to the Paley-Wiener theorem, the function $(g*\gp)(x)$
 has an analytic continuation
from the interval $[0,1]$ onto the whole complex plane,
see, e.g., \cite[Ex~27.2(3)]{T67}.
Then Eq.~\er{steqfi} takes the form
$g*\gp=C$ on $\C$. The function $g$, and then the function  $g*\gp$,
belongs to the Schwartz class. So, we have $(g*\gp)(x)\to 0$,
as $|x|\to\iy$.
Then $C=0$ in \er{steqfi} and we obtain the equation $g*\gp=0$ on $\R$.
Taking the Fourier transform (in the sense distributions), we obtain
$\wh g(\x)\wh \gp(\x)=0$, where $\wh g$ and $\wh \gp$ are Fourier transforms
of $g$ and $\gp$, respectively, $\wh \gp$ is an entire function,
see \cite[Th~7.19]{D13},
$\wh g(\x)>0$ for all $\x\in\R$. This yields $\wh \gp=0$
and $\gp=0$.~\BBox

\subsection{Equation for $\gp$}
Lemma~\ref{Lmstsolfi} shows that Eq.~\er{eqcVfipr}
has no classical stationary solution.
Later, we will show that a weak stationary solution exists. It will
be convenient for us to switch to an equation for the distribution density $\gp$
 defined in \er{defgpg*f} and \er{defgpZ1}.
In the next lemma, we derive the corresponding equation.

\begin{lemma}
\lb{Lmglclsol}
The function $\gp$ satisfies the equation \er{eqgpwt}.
\end{lemma}

\no {\bf Proof.}
Differentiating the expression for $p_w$ from \er{defgpZ1}
and substituting the result into \er{defgpg*f}, we obtain
\[
\lb{difgpwt}
{d\gp(x,\theta)\/d\theta}
=-\a\gp(x,\theta){d\cV(x,\theta )\/d\theta }
-{\gp(x,\theta)\/Z_w(\tau(\theta))}{dZ_w(\tau(\theta) )\/d\theta }
=-\a\gp(x,\theta)\Big({d\cV(x,\theta )\/d\theta }
-\Big\langle{d\cV\/d\theta }\Big\rangle_\gp(\theta )\Big).
\]
The identity \er{eqcVfipr} gives
$$
\Big\langle{d\cV\/d\theta }\Big\rangle_\gp(\theta )
=\int_0^1\big( G\star \gp(x,\theta)
-\langle G\star \gp \rangle(\theta)\big)\gp(x,\theta )dx
=\langle G\star \gp \rangle_\gp(\theta)-\langle G\star \gp \rangle(\theta),
$$
then
$$
{d\cV(x,\theta )\/d\theta }-\Big\langle{d\cV\/d\theta }\Big\rangle_\gp(\theta )
=G\star \gp(x,\theta)-\langle G\star \gp \rangle_\gp(\theta).
$$
Substituting this identity into \er{difgpwt}, we obtain \er{eqgpwt}.~\BBox

\medskip

Eq.~\er{eqgpwt} has a global solution.

\begin{lemma}
Let $\gp_0\in\cP[0,1]$, where $\cP[0,1]$ has the form \er{cP01}.
Then the solution $\gp(\cdot,\theta)\in\cP[0,1]$
to Eq.~\er{eqgpwt}, satisfying the condition $\gp(x,0)=\gp_0(x),x\in[0,1]$,
exists and is unique for all $\theta>0$.
\end{lemma}

\no {\bf Proof.} The result follows from Lemma~\ref{LmeucV}~ii).~\BBox

\subsection{Lyapunov function}

Define a positive-definite functional $\cE$
on the set $\cP[0,1]$, given by \er{cP01},
by the identity
\[
\lb{defcE}
\cE[ p ]={1\/2}\int_0^1\int_0^1g(x-y) p (x) p (y)dxdy.
\]
We show that the functional $\cE$
has all the properties of the Lyapunov function for Eq.~\er{eqgpwt}.

Introduce the functional $\cD$ on $\cP[0,1]$ by the identity
\[
\lb{defcD}
\cD[p]=\int_0^1p(x)\big(g\star p(x)-\langle g\star p\rangle_p\big)^2dx
=\langle (g\star p)^2\rangle_p-\langle g\star p\rangle_p^2,
\]
$\cD$ is the variance of $g \star p$ with respect to the measure $p$.
The definition gives that $\cD[p]\ge 0$ for all $p\in\cP[0,1]$.

\begin{lemma}
\lb{LmLF}
Let $\gp\in C^1(\R_+,\cP[0,1])$ be a solution to Eq.~\er{eqgpwt}. Then

i)
\[
\lb{derLF}
{d\cE[\gp (\cdot,\theta)]\/d\theta}
=-\a\cD[\gp(\cdot,\theta)]\le 0.
\]
In particular, the function $\cE[\gp (\cdot,\theta)]$ is monotonously decreasing
and has non-negative limit at $\theta\to\iy$.

ii) The functional $\cD$ satisfies
\[
\lb{varcD}
\int_0^\iy\cD[\gp(\cdot,\theta)]d\tau<\iy.
\]
\end{lemma}

\no {\bf Proof.} i) The definition \er{defcE} gives
$$
\begin{aligned}
{\cE[\gp ](\theta)\/d\theta}
={1\/2}\int_0^1\int_0^1g(x-y)\Big({\pa\gp(x,\theta)\/\pa\theta}\gp(y,\theta)
+\gp(x,\theta){\pa\gp(y,\theta)\/\pa\theta}\Big)dxdy
\\
=\int_0^1dx{\pa\gp(x,\theta)\/\pa\theta}\int_0^1g(x-y)\gp(y,\theta)dy
=\int_0^1{\pa\gp(x,\theta)\/\pa\theta}G\star \gp(x,\theta)dx.
\end{aligned}
$$
Substituting \er{eqgpwt}, we obtain
$$
{\cE[\gp ](\theta)\/d\theta}
=-\a\int_0^1\gp(x,\theta)
\big(G\star \gp(x,\theta)-\langle G\star \gp \rangle_\gp(\theta)\big)G\star \gp(x,\theta)dx.
$$
Using the identity
$$
\int_0^1\gp(x,\theta)
\big(G\star \gp(x,\theta)-\langle G\star \gp \rangle_\gp(\theta)\big)
\langle G\star \gp \rangle_\gp(\theta)dx=0,
$$
we obtain \er{derLF}.

ii) By integrating \er{derLF}, we obtain
$$
\int_0^\iy\cD[\gp(\cdot,\theta)]d\theta
={1\/\a}\big(\cE[\gp (\cdot,0)]-\cE[\gp (\cdot,\iy)]\big),
$$
which yields \er{varcD}.~\BBox

\subsection{Dynamics of solutions}

%1. Since $\cD[\gp(\cdot,\theta)]$ is continuously differentiable,
%it follows from \er{varcD} that $\cD[\gp(\cdot,\theta)]\to 0$, as $\theta\to+\iy$.
%The definition \er{defcD} implies
%$$
%\cD[\gp(\cdot,\theta)]=\int_0^1\gp(x,\theta)\big(G\star \gp(x,\theta)
%-\langle G\star \gp(\theta)\rangle_{\gp(\cdot,\theta)}\big)^2dx
%\to 0,\qq
%\theta\to+\iy.
%$$
%This yields that
%$g\star \gp(x,\theta)\to \langle g\star \gp(\cdot,\theta)\rangle_{\gp(\cdot,\theta)}$
%for all $x\in[0,1]$ such that $\gp(x,\theta)$ is greater than $0$.
%
%2. \er{derLF} yields that ${d\cE[\gp (\cdot,\theta)]\/d\theta}<0$.
%Therefore, the energy
%$\cE[\gp (\cdot,\theta)]$ decreases monotonically as $\theta$ increases.
%This guarantees the absence of cycles and chaotic movement.

In the following lemma,
we prove that the system \er{eqgpwt} cannot exhibit complex dynamics.

\begin{proposition}
\lb{Lmdiss}
 i) The solution $\gp(\cdot,\theta)$ to Eq.~\er{eqgpwt}
 cannot be periodic, in the sense that
$\gp(\cdot,\theta+\gT)\ne\gp(\cdot,\theta)$
for any $\gT>0$.

ii) The solution $\gp(\cdot,\theta)$ to Eq.~\er{eqgpwt}
cannot return to its original value in the sense that
$\lim_{n\to\iy}\gp(\cdot,\theta+\theta_n)\ne\gp(\cdot,\theta)$
for any subsequence $(\theta_n)_{n\in\N}$, tending to infinity.

\end{proposition}

\no {\bf Proof.}
i) Let $\gp(\cdot,\theta+\gT)=\gp(\cdot,\theta)$
for some $\gT>0$. Then $\cE[\gp(\cdot,\theta+\gT)]=\cE[\gp(\cdot,\theta)]$.
Since ${d\cE[\gp (\cdot,\theta)]\/d\theta}\le 0$, we have
$\cE[\gp(\cdot,\theta)]=\const$ on the interval $[\theta,\theta+\gT]$.
Moreover, the identity \er{derLF} gives
$\cD[\gp(\cdot,\theta)]=0$ on this interval, and hence $\gp(x,\theta)=1$
for all $(x,\theta)\in (0,1)\ts[\theta,\theta+\gT]$.
Lemma~\er{Lmstsolfi} shows that the function $\gp(x,\theta)=1$
cannot satisfy Eq.~\er{eqgpwt}, which yields the statement.

ii) Let $\lim_{n\to\iy}\gp(\cdot,\theta+\theta_n)=\gp(\cdot,\theta)$,
$\lim_{n\to\iy}\theta_n=+\iy$.
Then $\lim_{n\to\iy}\cE[\gp(\cdot,\theta+\theta_n)]=\cE[\gp(\cdot,\theta)]$.
Since ${d\cE[\gp (\cdot,\theta)]\/d\theta}\le 0$, we have
$\cE[\gp(\cdot,\theta)]=\const$ on the interval $[\theta,+\iy)$.
Moreover, the identity  \er{derLF} gives
$\cD[\gp(\cdot,\theta)]=0$ on this interval, and hence $\gp(x,\theta)=1$
for all $(x,\theta)\in (0,1)\ts[\theta,+\iy)$.
Lemma~\er{Lmstsolfi} shows that the function $\gp(x,\theta)=1$
cannot satisfy Eq.~\er{eqgpwt}, which yields the statement.~\BBox

\subsection{$\gp=1$ as a quasi-stationary distribution }
The following proposition demonstrates that the state with $\gp=1$
is quasi-stationary for small $\sigma$ in the following sense:
the change in the distribution within a $\sigma$-neighborhood
of the edges occurs at a rate of order unity, whereas at a distance
greater than $\sigma$ from the edges,
the change becomes slow--of order $\sigma$.

\begin{proposition}
\lb{Prop1qs}
 Let $\gp\in C^1(\R_+,\cP[0,1])$
be the solution to Eq.~\er{eqgpwt}, satisfying
the condition $\gp(x,0)=1$ for all $x\in[0,1]$.
Then for any fixed $M>0$ and for all $\s>0$ small enough
the following estimates hold true:
\[
\lb{asdergpp=1edge}
{d\gp(x,\theta)\/d\theta}\big|_{\theta=0}
=-2\a\F\Big({x\/\s}\Big)+\ve_1,\qq x\in[0,M\s],
\]
\[
\lb{asdergpp=1edge+}
{d\gp(x,\theta)\/d\theta}\big|_{\theta=0}
=2\a\F\Big({1-x\/\s}\Big)+\ve_2,\qq x\in[1-M\s,1],
\]
\[
\lb{asdergpp=1centre}
\Big|{d\gp(x,\theta)\/d\theta}\big|_{\theta=0}\Big|
\le\a\big(2\F(M)+\s\ve_0\big),\qq
x\in[M\s,1-M\s],
\]
where
$$
\F(x)={1\/\sqrt{2\pi}}\int_x^\iy e^{-{s^2\/2}}ds,
\qq
\ve_0=\sqrt{2\/\pi}+2\int_{1\/\s}^\iy\F(x)dx,\qq
\max\{|\ve_1|,|\ve_2|\}\le\a\Big(2\F\Big({1\/\s}-M\Big)
+\s\ve_0\Big).
$$

%
%\[
%\lb{derngp=1at0}
%\|\pa_\theta\gp(\cdot,0)\|_{L^2}^2=\a^2{\sqrt2-1\/\sqrt\pi}\s+O(\s^2),
%\]
%\[
%\lb{der1ngp=1at0}
%\|\pa_\theta\gp(\cdot,0)\|_{L^1}=O(\s),
%\]
%\[
%\lb{dergp=1at0}
%\pa_\theta\gp(0,0)={\a\/2}+o(1),
%\]
%as $\s\to 0$.
\end{proposition}

\no {\bf Proof.} Eq.~\er{eqgpwt} implies
\[
\lb{dergpatp=1}
{d\gp(x,\theta)\/d\theta}\Big|_{\theta=0}
=-\a\big((g\star 1)(x)-\int_0^1(g\star 1)(x)dx\big).
\]
The identity
\[
\lb{gF-F}
(g\star 1)(x)=\F\Big({x\/\s}\Big)-\F\Big({1-x\/\s}\Big),
\]
gives
$$
\int_0^1(g\star 1)(x)dx=2\s\int_0^{1\/\s}\F(x)dx.
$$
Using the identity
$$
\int_0^\iy\F(x)dx={1\/\sqrt{2\pi}}\int_0^\iy x e^{-{x^2\/2}}dx
={1\/\sqrt{2\pi}},
$$
we obtain
\[
\lb{asmconj}
\int_0^1(g\star 1)(x)dx
=2\s\Big({1\/\sqrt{2\pi}}-\int_{1\/\s}^\iy\F(x)dx\Big).
%=\sqrt{2\/\pi}\s\big(1+o(1)\big).
\]
If $x\in[M\s,1-M\s]$ for some $M>0$ small enough,
then the identity \er{gF-F} gives
$$
(g\star 1)(x)\le2\F(M).
$$
Substituting this estimate and the identity \er{asmconj} into
\er{dergpatp=1}, we obtain \er{asdergpp=1centre}.

Let $x=\x\s\in[0,M\s]$. Then the identity \er{gF-F} implies
$$
(g\star 1)(x)=\F(\x)-\F\Big({1\/\s}-\x\Big).
$$
Substituting this estimate and the identity \er{asmconj} into
\er{dergpatp=1}, we obtain \er{asdergpp=1edge}.
The proof of \er{asdergpp=1edge+} is similar.~\BBox

\medskip

\no {\bf Remark.}
1) Since we can choose $M$ to be sufficiently large
in the estimate \er{asdergpp=1centre},
the derivative ${d\gp\/d\theta}|_{\theta=0}$
outside the neighborhood of the edges is of the order of $\s$.
At the same time, estimates \er{asdergpp=1edge} and
\er{asdergpp=1edge+} show that the rate of change
of $\gp$ in the vicinity of the edges is proportional to a constant.

2) We have $g(x)=\d(x)$ in the limit $\s=0$. In this limit the solution
$\gp=1$ is a stationary solution and
$\lim_{\theta\to+\iy}\gp(x,\theta)=1$ for all solutions to Eq.~\er{eqgpwt}.

%???? This yields
%$$
%\cD[1]={\sqrt2-1\/\sqrt\pi}\s+O(\s^2).
%$$
%The identity \er{derLF} gives
%%\er{derLFgp=1at0}.
%\[
%\lb{derLFgp=1at0}
%{d\cE[\gp(\cdot,\theta)]\/d\theta}\Big|_{\theta=0}=-\a{\sqrt2-1\/\sqrt\pi}\s+O(\s^2),
%\]
%Moreover, using
%$\|\pa_\theta\gp(\cdot,0)\|_{L^2}^2=a^2\cD[1]$,
%we obtain \er{derngp=1at0}.
%
%??????

\section{Finite interval, week solutions}
\setcounter{equation}{0}

\subsection{Week solutions}
Let $\cM[0,1]$ be the space of probability measures on the interval  $[0,1]$.
For $\m\in\cM[0,1]$ we define the energy functional
\[
\lb{defcEf}
\cE[\m]={1\/2}\int_{[0,1]^2}g(x-y)\m(dx)\m(dy)
={1\/2}\int_{[0,1]}U_\m(x)\m(dx)
\]
and the dissipation
\[
\lb{defcDf}
\cD[\m]=\int_0^1\big(U_\m(x)-\langle U_\m\rangle_{\m}\big)^2\m(dx).
\]

\begin{lemma}
\lb{Lmgp}

i) The mapping $\m\mapsto U_\m$ is a continuous mapping
from $\cM[0,1]$ to $C([0,1])$, that is,
the week-* convergence $\m_n\rightharpoonup\m$ implies
$\sup_{x\in[0,1]}|U_\m(x)-U_{\m_n}(x)|\to 0$.

ii) The functionals $\cE$ and $\cD$ are continuous
in the  $\cM[0,1]$-metric,
that is, the week-* convergence $\m_n\rightharpoonup\m$  implies $\cE[\m_n]\to\cE[\m]$
and $\cD[\m_n]\to\cD[\m]$.
The set of values of the functionals
$\cE[\mu]$ and $\cD[\mu],\mu\in\cM[0,1]$, is compact.
Moreover, if the family of measures $\m_\theta^{(n)}$
converges in $\cM[0,1]$ to the family of measures $\m_\theta$ uniformly in $\theta$
on a certain compact set $[0,\cT],\cT>0$, then $\cE[\m_\theta^{(n)}]$ and
$\cD[\m_\theta^{(n)}]$
converge to $\cE[\m_\theta]$ and $\cD[\m_\theta]$, respectively,
uniformly on that compact set.

\end{lemma}

\no {\bf Proof.}
i) The definitions \er{wcoonv} and \er{gp} give
$$
|U_{\m}(x)-U_{\m_n}(x)|
=\Big|\int_{[0,1]}g(x-s)\big(\m(ds)-\m_n(ds)\big)\Big|\to 0.
$$
Thus, we have the pointwise convergence  $U_{\m_n}(x)\to U_{\m}(x)$
for all $x\in[0,1]$.
Moreover,
$$
\sup_{x\in[0,1]}|U_{\m_n}(x)|\le\sup_{x\in\R}|g(x)|.
$$
Then the family $\{U_{\m_n}(x)\}$ is uniformly bounded.
Furthermore, the estimate
$$
|g(x-s)-g(y-s)|\le \o(|x-y|)\qq \forall\ \ x,y,s\in[0,1]
$$
gives
$$
|U_{\m_n}(x)-U_{\m_n}(y)|=|\int_0^1(g(x-s)-g(y-s))\m_n(ds)|\le\o(|x-y|),
$$
where $\o$ is the modulus of continuity of the function $g$ on $[-1,1]$.
Thus, the family $\{U_{\m_n}(x)\}$ is equicontinuous.
%The Arzela-Ascoli theorem gives that this set is compact.

Let $\ve>0$. Let $\d>0$ be such that $|x-y|<\d$ implies
$|U_{\m_n}(x)-U_{\m_n}(y)|\le{\ve\/3}$ and
$|U_{\m}(x)-U_{\m}(y)|\le{\ve\/3}$.
Let $x_1,...,x_K$ be such that $|x-x_j|<\d$ for all $x\in[0,1]$
and for some $j=1,...,K$. Then $|U_{\m_n}(x_j)-U_{\m}(x_j)|\le{\ve\/3}$
for all $n>N$, where $N\in\N$ is large enough.
Let $x\in[0,1]$, let $x_j$ be the nearest to $x$ point, and let $n>N$.
Then
$$
|U_{\m_n}(x)-U_{\m}(x)|\le
|U_{\m_n}(x)-U_{\m_n}(x_j)|+|U_{\m_n}(x_j)-U_{\m}(x_j)|
+|U_{\m}(x_j)-U_{\m}(x)|\le\ve.
$$
This yields the statement.

ii) Let $\m_n\rightharpoonup\m$. Then
the statement~i) gives $\sup_{x\in[0,1]}|U_{\m}(x)-U_{\m_n}(x)|\to 0$,
which yields $\langle U_{\m_n}\rangle_{\m_n}\to\langle U_\m\rangle_{\m}$.
Therefore, $f_n(x)=U_{\m_n}(x)-\langle U_{\m_n}\rangle_{\m_n}$
converges to $f(x)=U_{\m}(x)-\langle U_{\m}\rangle_{\m}$
uniformly in $x\in[0,1]$.
Then the definitions
\er{defcEf} and \er{defcDf} imply
$$
\cE[\m_n]={1\/2}\int_{[0,1]}U_{\m_n}(x)\m_n(dx)\to
{1\/2}\int_{[0,1]}U_{\m}(x)\m(dx)=\cE[\m],
$$
$$
\cD[\m_n]=\int_0^1f_n^2(x)\m_n(dx)\to \int_0^1f^2(x)\m(dx)=\cD[\m],
$$
that is, the functionals $\cE$ and $\cD$
are continuous in the $\cM[0,1]$-metric.
By Prokhorov's theorem, the set $\cM[0,1]$ is weakly-* compact, that is,
each sequence $(\m_n)\ss\cM[0,1]$ contains a weakly-*
convergent subsequence.
Consequently, each numerical sequence
$\cE[\mu_n])$ and $\cD[\mu_n]$ contains a convergent subsequence,
then the set of values of the functionals
$\cE[\mu]$ and $\cD[\mu],\mu\in\cM[0,1]$, is compact.

Moreover, if the family $\m_\theta^{(n)}$
converges in $\cM[0,1]$ to the family $\m_\theta$ uniformly in $\theta$
on a  compact set $[0,\cT],\cT>0$, then
$$
\sup_{\theta\in[0,\cT]}|\cE[\m_\theta^{(n)}]-\cE[\m_\theta]|\to 0,
$$
$$
\sup_{\theta\in[0,\cT]}|\cD[\m_\theta^{(n)}]-\cD[\m_\theta]|\to 0.
$$
Therefore, $\cE[\m_\theta^{(n)}]$ and
$\cD[\m_\theta^{(n)}]$
converge to $\cE[\m_\theta]$ and $\cD[\m_\theta]$, respectively,
uniformly on $[0,\cT]$.~\BBox

\medskip

The functionals \er{defcEf} and \er{defcDf} are extensions
of the functionals \er{defcE} and \er{defcD}, respectively.

\begin{lemma}
\lb{Lmfsmd}
Any weak stationary measure from $\cM[0,1]$
is a purely atomic measure with no accumulation points, that is,
\[
\lb{stmdiskr}
\m_*=\sum_{n=1}^Na_n\d_{x_n},\qq a_n>0,\qq 0=x_1<x_2<...<x_N=1,\qq n=1,...,N,
\]
for some $N\in\N$.
\end{lemma}

\no {\bf Proof.} Let $\m\in\cM[0,1]$ be a weak stationary measure.
The identity \er{steqfif} and the continuity $U_\m$ give
$U_{\m}(x)=c$ on $\supp\m$,
where $c=\langle U_{\m}\rangle_{\m}$. The function $U_\m-c$
has an analytic continuation to the non-trivial entire function,
(see the arguments from the proof of Lemma~\ref{Lmstsolfi}),
so the set of its zeros on the interval $[0,1]$
is finite. Thus, $\supp\m$ is a finite set, which yields
\er{stmdiskr}.~\BBox

\medskip

\no {\bf Remark.} Weak stationary measures with the same energy $\cE$ can form continuous families.
For example, atomic measures $\d_a$, $a\in[0,1]$,
form a continuous family
of weak stationary measures with the energy $\cE[\d_a]={1\/2}g(0)$.
For any fixed $b\in(0,1]$
diatomic measures ${1\/2}(\d_a+\d_{a+b})$
with $a\in[0,1-b]$ are weekly stationary and has the same energy
$\cE[{1\/2}(\d_a+\d_{a+b})]={1\/4}(g(0)+g(b))$.

\subsection{Dynamics of the weak solution}

It follows from the next theorem and the preceding results
that the system is purely dissipative in nature,
exhibits no stable cycles, and its long-term behavior is entirely
determined by the $\omega$-limit set. The trajectory cannot
contain non-trivial periodic orbits; the $\omega$-limit set
consists of weakly stationary measures; and the asymptotic dynamics
are non-oscillatory: the solution does not wander between different regimes
but relaxes to one of the limit sets of weakly stationary measures with
constant energy. We recall that, by Lemma~\ref{Lmfsmd}, any weakly
stationary measure is a purely atomic measure with no accumulation points.

\begin{theorem}
\lb{Thfs}
Let $\m_0\in\cM[0,1]$. Then there exists the global weak solution
$(\m_\theta)_{\theta\ge 0}\ss\cM[0,1]$
to Eq.~\er{eqgpwtff},
satisfying the initial condition $\m_\theta|_{\theta=0}=\m_0$. Moreover,
this solution is continuous with respect to the initial conditions,
that is, if $\n_\theta$ is the solution, satisfying
the initial condition $\n_\theta|_{\theta=0}=\n_0$, then
\[
\lb{estdifmn}
\|\m_\theta-\n_\theta\|_{TV}\le e^{5\a\/\sqrt{2\pi}\s}\|\m_0-\n_0\|_{TV},
\]
where $\|\m-\n\|_{TV}$ is the total variation norm, given by
$$
\|\l\|_{TV}=\l_+([0,1])+\l_-([0,1]),
$$
$\l$ is a finite signed measure on $[0,1]$, $\l=\l_+-\l_-$ is the Jordan
decomposition of the measure $\l$ into
two orthogonal non-negative measures
$\l_+$ and $\l_-$.

Furthermore,

i) The function $\theta\mapsto\cE[\m_{\theta}]$, given by \er{defcEf},
satisfies the identity
\[
\lb{derLFf}
\cE[\m_t]+\a\int_s^t\cD[\m_u]du
=\cE[\m_s],\qq 0\le s\le t,
\]
where $\cD$ is given by \er{defcDf}. In particular,
the function $\theta\mapsto\cE[\m_{\theta}]$ is non-increasing
and has a finite limit.

ii) The function $\theta\mapsto\cD[\m_{\theta}]$, given by \er{defcDf},
satisfies the estimate
\[
\lb{varcDf}
\int_0^\iy\cD[\m_{\theta}]d\theta<\iy.
\]
In particular, $\theta\mapsto\cD[\m_{\theta}]$ tends to zero, at least along
subsequences $\theta_n\to\iy$.

iii) Let $\o(\m_0)$ be the set of the limit points of the trajectory
$\m_{\theta}$,
starting at the point $\m_0$, $\o(\m_0)$ consists of points $\m_\iy$ of the set
$\cM[0,1]$ such that $\m_{t_n}\rightharpoonup\m_\iy$
along some sequence $t_n\to+\iy$. Then

a) Each point $\m_\iy$ of the set $\o(\m_0)$
is a week stationary mesure:
\[
\lb{lmfs}
\big(U_{\m_\iy}(x)-\langle U_{\m_\iy}\rangle_{\m_\iy}\big)\m_\iy=0,
\]

b) The set $\o(\m_0)$ consists of finitely atomic measures,

c) $\cE$ is constant on the set $\o(\m_0)$.

\end{theorem}

\no {\bf Proof.}
Choose a sequence of densities $\gp_0^{(n)}\in\cP[0,1]$
such that $\gp_0^{(n)}dx\to\m_0(dx)$ weekly-* in $\cM[0,1]$.
We proved in Lemma~\ref{Lmglclsol}~ii) that for each density
$\gp_0^{(n)}\in\cP[0,1]$ for all $\theta>0$ there exists a unique
solution $\gp_n(\cdot,\theta)\in\cP[0,1]$
to Eq.~\er{eqgpwt}, satisfying the condition
$\gp_n(x,0)=\gp_0^{(n)}(x),x\in[0,1]$.
Introduce the sequence of mesures $\m_\theta^{(n)}\in\cM[0,1]$ by
$\m_\theta^{(n)}(dx)=\gp_n(x,\theta)dx,x\in[0,1]$.

Let $\vp\in C[0,1]$. Then \er{eqgpwtf} gives
\[
\lb{deqfsol}
{d\/d\theta}\int_{[0,1]}\vp(x)\m_\theta^{(n)}(dx)
=-\a\int_{[0,1]}\vp(x)
\big(U_{\m_\theta^{(n)}}(x)-\langle U_{\m_\theta^{(n)}}\rangle_{\m_\theta^{(n)}}\big)
\m_\theta^{(n)}(dx).
\]
This yields
$$
\Big|{d\/d\theta}\int_{[0,1]}\vp(x)\m_\theta^{(n)}(dx)\Big|
\le\a\max_{x\in[0,1]}|\vp(x)|\max_{x\in[0,1]}
\big|U_{\m_\theta^{(n)}}(x)-\langle U_{\m_\theta^{(n)}}\rangle_{\m_\theta^{(n)}}\big|
\le 2\a\max_{x\in[0,1]}|\vp(x)|\max_{x\in[-1,1]}g(x).
$$
Since this estimate depends on neither $n$ nor $\theta$, for any
$\vp\in C[0,1]$, the family of functions
$\theta\mapsto\int_{[0,1]}\vp(x)\m_\theta^{(n)}(dx)$
is uniformly bounded and equicontinuous on any finite interval $[0,\cT]$.
By the Arzela--Ascoli theorem, there exists a subsequence
$\m_\theta^{(n_k)}\in\cM[0,1],\theta\in[0,\cT]$, converging weak-* uniformly
in $\theta\in[0,\cT]$ to some family $\m_\theta\in\cM[0,1]$, that is,
$$
\sup_{\theta\in[0,\cT]}\Big|\int_0^1\vp(x)\m_\theta^{(n_k)}(dx)
-\int_0^1\vp(x)\m_\theta(dx)\Big|\to 0,
$$
for all $\vp\in C[0,1]$.
This can be done simultaneously on all intervals $[0,\cT]$,
 $\cT=1,2,3,4,...$ using a diagonal process.
Lemma~\ref{Lmgp} gives
$U_{\m_n}\to U_\m$ uniformly in $x\in[0,1]$, then
$\langle U_{\m_n}\rangle_{\m_n}\to\langle U_{\m}\rangle_{\m}$.
Therefore, the sequence
$\int_{[0,1]}\vp(x)U_{\m_\theta^{(n)}}(x)\m_\theta^{(n)}(dx)$
converges to $\int_{[0,1]}\vp(x)U_{\m_\theta}(x)\m_\theta(dx)$
uniformly in $\theta\in[0,\cT]$.
Eq.~\er{deqfsol} gives
$$
\int_{[0,1]}\vp(x)\m_\theta^{(n)}(dx)-\int_{[0,1]}\vp(x)\m_0^{(n)}(dx)
=-\a\int_0^\theta\int_{[0,1]}\vp(x)
\big(U_{\m_t^{(n)}}(x)-\langle U_{\m_t^{(n)}}\rangle_{\m_t^{(n)}}\big)
\m_t^{(n)}(dx)dt.
$$
Passing to the limit on both sides, we obtain
$$
\int_{[0,1]}\vp(x)\m_\theta(dx)-\int_{[0,1]}\vp(x)\m_0(dx)
=-\a\int_0^\theta\int_{[0,1]}\vp(x)
\big(U_{\m_t}(x)-\langle U_{\m_t}\rangle_{\m_t}\big)
\m_t(dx)dt.
$$
This identity shows that $\m_\theta$
is a global weak solution to Eq.~\er{eqgpwtff}.

Eq.~\er{eqgpwtff} implies
\[
\lb{defFmu}
\m_\theta-\n_\theta=\int_0^\theta (F_{\m_s}-F_{\n_s})ds,\qq
F(\m)=-\a (U_{\m}-\langle U_{\m}\rangle_{\m})\m.
\]
The definitions \er{gp} yield
$$
U_{\m}(x)-U_{\n}(x)=\int_{[0,1]}g(x-s)(\m(ds)-\n(ds)),\qq x\in[0,1].
$$
Then
\[
\lb{Am-an}
\sup_{x\in[0,1]}|U_{\m}(x)-U_{\n}(x)|\le\max_{x\in\R}|g(x)|\|\m-\n\|_{TV}.
\]
Moreover,
$$
\langle U_{\m}\rangle_{\m}-\langle U_{\n}\rangle_{\n}
=\int_{[0,1]}U_\m(x)(\m(dx)-\n(dx))
+\int_{[0,1]}(U_\m(x)-U_\n(x))\n(dx),
$$
which yields
\[
\lb{cm-cn}
|\langle U_{\m}\rangle_{\m}-\langle U_{\n}\rangle_{\n}|
\le\sup_{x\in[0,1]}|U_\m(x)|\|\m-\n\|_{TV}
+\max_{x\in\R}|g(x)|\|\m-\n\|_{TV}\int_{[0,1]}\n(dx)
\le 2\max_{x\in\R}|g(x)|\|\m-\n\|_{TV}.
\]
Furthermore,
\[
\lb{Am-Amm}
\sup_{x\in[0,1]}|U_\m(x)-\langle U_{\m}\rangle_{\m}|
\le 2\max_{x\in\R}|g(x)|.
\]
The definition \er{defFmu} yields
$$
F(\m)-F(\n)
=-\a(U_{\m}-\langle U_{\m}\rangle_{\m})(\m-\n)
-\a(U_{\m}-U_{\n}+\langle U_{\n}\rangle_{\n}-\langle U_{\m}\rangle_{\m})\n.
$$
The estimates \er{cm-cn}, \er{Am-an}, and \er{Am-Amm} give
$$
\|F(\m)-F(\n)\|_{TV}\le5\a\max_{x\in\R}|g(x)|\|\m-\n\|_{TV}.
$$
Substituting this estimate into the identity \er{defFmu}
and using the estimate
%$$
%\|F_{\m_s}-F_{\n_s}\|_{TV}\le 5\max_{x\in\R}|g(x)|\|\m_0-\n_0\|_{TV},
%$$
%and
$\max_{x\in\R}|g(x)|\le{1\/\sqrt{2\pi}\s}$,
we obtain
$$
\|\m_\theta-\n_\theta\|_{TV}\le{5\a \/\sqrt{2\pi}\s}\int_0^\theta\|\m_s-\n_s\|_{TV},
$$
which yields \er{estdifmn}. The uniqueness follows.

i), ii) The identity \er{derLF} gives
$$
{d\cE[\m_\theta^{(n)}]\/d\theta}
=-\a\cD[\m_\theta^{(n)}],\qq \theta\ge 0,
$$
which yields
$$
\cE[\m_t^{(n)}]+\a\int_s^t\cD[\m_u^{(n)}]du
=\cE[\m_s^{(n)}],\qq 0\le s\le t.
$$
It follows from Lemma~\ref{Lmgp}~ii) that one can pass to the limit,
yielding \er{derLFf}.
Eq.~\er{derLFf} implies energy decay and the integrability of dissipation:
$$
\cE[\m_t]\le\cE[\m_s],\qq \a\int_s^t\cD[\m_u]du=\cE[\m_s]-\cE[\m_t].
$$
Passing to the limit $s\to 0,t\to+\iy$, we obtain \er{varcDf}.

iii) a) Let $\m_n=\m_{\theta_n}\to\m_\iy$ be a sequence such that
$\cD[\m_n]\to 0$. Due to the continuity of the kernel on a compact set, the weak convergence of the sequence $\mu_n$
implies the uniform convergence of the sequence $U_{\mu_n}$:
$U_{\m_n}\to U_{\m_\iy}$ uniformly in $x\in[0,1]$.
Therefore, $U_{\m_n}-\langle U_{\m_n}\rangle_{\m_n}$ converges to
$U_{\m_\iy}-\langle U_{\m_\iy}\rangle_{\m_\iy}$ uniformly in $x\in[0,1]$.
This yields $\cD[\m_{n}]\to\cD[\m_\iy]$ and, due to $\cD[\m_n]\to 0$,
we obtain
$$
\cD[\m_\iy]
=\int_0^1\big(U_{\m_\iy}(x)-\langle U_{\m_\iy}\rangle_{\m_\iy}\big)^2\m_\iy(dx)
=0.
$$
This yields \er{lmfs}.

b) It follows from Lemma~\ref{Lmfsmd} that the $\omega$-limit set
consists of finitely atomic measures.

c) Due to i), there exists $\lim_{\theta\to\iy}\cE[\m_{\theta}]=\cE_\iy$.
If $\m_{\iy}\in\o(\m_0)$, then there exists the sequence
$\m_{t_n}\rightharpoonup\m_\iy$. Then
$\cE[\m_{t_n}]\to\cE[\m_\iy]$. This yields $\cE[\m_\iy]=\cE_\iy$
for all $\m_{\iy}\in\o(\m_0)$. Thus,
$\cE$ is constant on the set $\o(\m_0)$ and equals to $\cE_\iy$.~\BBox

\medskip

\no {\bf Remark.}
Thus, the $\omega$-limit set $\omega(\mu_0)$ may contain
several limit points. The following scenario then arises:
for large $\theta$, the trajectory becomes very slow
(with dissipation $D[\mu_\theta] \to 0$);
it remains close to the set of stationary points
at all times but may move very slowly between different
configurations without approaching a limit. Formally,
$\text{dist}(\mu_\theta, \omega(\mu_0)) \to 0$, yet $\mu_\theta$ does not converge.

\medskip

\no {\bf Proof of Theorem~\ref{Thfiws}.}
The results follow from Theorem~\ref{Thfs}.~\BBox

\subsection{Relationship between the times}

\begin{lemma}
\lb{Lmrshtim}
 The function $\theta$, defined by \er{thetafi}, is continuously
differentiable and strictly increasing.
Moreover, there exists $\tau_*>0$ such that
$\theta:[0,\tau_*)\to [0,+\iy)$.
\end{lemma}

\no {\bf Proof.} The functions $Z_w$ and $Z_b$ are continuous and positive,
then $\theta$ is continuously differentiable and strictly increasing.
Theorem~\ref{Thfiws} implies that the weak solution $\m_\theta$
goes to a finitely atomic measure $\m_\iy=\sum_{j=1}^N a_j\d_{x_j}$
for some $0\le x_1\le ...\le x_N$, $a_j>0$, $N\in\N$.
This yields $G\star\gp(x,\theta)=H(x)+o(1)$
as $\theta\to+\iy$, where $H(x)=\sum_{j=1}^N a_jG(x,x_j)$.
Eq.~\er{eqcVfipr} implies
$$
\wt V(x,\tau)=\cV(x,\theta )=(H(x)-\langle H\rangle)\theta+o(\theta),\qq
\theta\to+\iy.
$$
Substituting this asymptotics into the definitions \er{defga}
 and \er{defgpg*f} and using the standard Laplace asymptotics we obtain
$$
\log\cZ_w(\theta)=\a(\langle H\rangle-h_{min})\theta+O(\log\theta),
\qq
\log\cZ_b(\theta)=\b(\langle H\rangle-h_{min})\theta+O(\log\theta),
$$
as $\theta\to+\iy$,
where $\cZ_b(\theta)=Z_b(\tau)$,
$h_{min}=\inf_{[0,1]} H(x)<\langle H\rangle$.
Therefore,
$$
{d\tau\/d\theta}={\cZ_b(\theta)\/\cZ_w(\theta)}
=e^{-\b(\langle H\rangle-h_{min})\theta}(1+o(1)),
\qq\theta\to+\iy.
$$
This function is integrable on $\R_+$, then $\tau(\theta)\to\tau_*<\iy$,
as $\theta\to+\iy$, which proves the statement.~\BBox

\medskip

\no {\bf Proof of Theorem~\ref{Thficl}~ii).} The result follows from
Lemma~\ref{Lmrshtim} and Theorem~\ref{Thfiws}.~\BBox

\subsection{The solution in the class $L^1(0,1)$}

We prove Theorem~\ref{ThL1}.
\medskip

\no {\bf Proof of Theorem~\ref{ThL1}.}
Let $\mu_\theta$ satisfy Eq.~\er{eqgpwtff}
and the initial condition $\mu_\theta|_{\theta=0}=\mu_0$.
Introduce the function
$$
b(x,\theta):=-\a\bigl(U_{\mu_\theta}(x)-c(\mu_\theta)\bigr),\qq
(x,\theta)\in[0,1]\ts[0,\infty),
\qq
c(\mu_\theta)=\langle U_{\mu_\theta}\rangle_{\mu_\theta}.
$$
The definitions \er{gp} give
$$
\sup_{x\in[0,1]}|U_{\mu_\theta }(x)|\le \sup_{x\in\R}|g(x)|,
\qquad
|c(\mu_\theta )|\le \sup_{x\in\R}|g(x)|,
$$
then
\[
\lb{estbxgt}
 \sup_{x\in[0,1]}|b(x,\theta)|\le 2\a\sup_{x\in\R}|g(x)|.
\]
The function $\theta \mapsto\mu_\theta $ is weekly continuous, then
$U_{\mu_\theta }\in C\bigl([0,\infty),C([0,1])\bigr)$.
This yields
$b\in C\bigl([0,\infty),C([0,1])\bigr)$.
Consider the function
\[
\lb{defRmes}
R(x,\theta )=e^{\int_0^\theta  b(x,s)ds}
\]
and the measure
$$
\nu_\theta (dx)=R(x,\theta )\mu_0(dx).
$$
Immediate differentiation shows that $\nu_\theta $ satisfies the linear equation
$$
{\partial \nu_\theta\/\partial\theta} =b(\cdot,\theta )\nu_\theta ,
\qquad
\nu_0=\mu_0,\qq\theta>0.
$$
On the other hand, for fixed function $b$ the initial measure $\mu_\theta $
satisfies the same linear equation.
Due to Theorem~\ref{Thfs},
this solution is unique in the space of finite measures.
Therefore,
$$
\mu_\theta =\nu_\theta =R(\cdot,\theta )\mu_0,\qq\theta>0.
$$
Since $\mu_0(dx)=\gp_0(x)dx$, we obtain
$$
\mu_\theta (dx)=R(x,\theta )\gp_0(x)dx,\qq\theta>0.
$$
Therefore,
\[
\lb{idpp0}
\gp(x,\theta )=R(x,\theta )\gp_0(x),\qq (x,\theta)\in[0,1]\ts\R_+.
\]
Moreover, $R$ is continuously differentiable and bounded in $t$
on any finite interval, then
$$
{\pa \gp(\cdot ,\theta )\/\pa \theta}=
b(\cdot,\theta )\gp(\cdot ,\theta )
$$
in $L^1(0,1)$, that is,
%. After the substitution$$\mu_\theta (dx)=\gp(x,\theta )dx$$we obtain that
$\gp$ satisfies Eq.~\er{eqgpwt}.
%Then
%$$
%\gp(x,\theta)=\gp_0(x)e^{R(x,\theta )}.
%$$
Substituting the estimate \er{estbxgt} into the definition \er{defRmes},
we obtain
$$
\sup_{x\in[0,1]}|R(x,\theta )|\le e^{ 2\a\sup_{x\in\R}|g(x)|\theta}
\le e^{\sqrt{2\/\pi} {\a\/\s}\theta}.
$$
The identity \er{idpp0} gives the estimate \er{2estp}
and then the relations \er{p-p_0=0}.~\BBox

\section{ Scheme ``INTERVAL''}
\setcounter{equation}{0}

\subsection{Differential equation}
\lb{SectDESI}

The basic equation takes the form
\[
\lb{ieii}
V_{n+1}(x)=V_n(x)+\mG(x,x_{n+1})e^{-{V_n(x_{n+1})\/\D T}},\qq
x\in\R,\qq V_0=0,
\]
for all $n\ge 0$, where
$$
\mG(x,s)=\ca hg(x-s),&x\in[0,1]\\
hg(-s),&x<0\\
hg(1-s),&x>1\ac,\qq x,s\in\R.
$$
Thus, $V_n\in C(\R)$ for all $n\ge 0$ and
\[
\lb{ieii2}
V_{n+1}(x)=V_n(x)+he^{-{V_n(x_{n+1})\/\D T}}
 g(x-x_{n+1}),\qq
x\in[0,1],
\]
$$
V_n(x)=\ca V_n(0)=\const,&x<0\\V_n(1)=\const,&x>1\ac .
$$
Introduce the mean value $\ol V_n=\int_0^1 V_n(x)dx$ and the function
$$
\wt V_n(x)=V_n(x)-\ol V_n,\qq n\ge 1,\qq x\in\R.
$$
Eqs~\er{ieii2} yield that $\wt V_n\in C(\R)$ for all $n\ge 0$ and
\[
\ol V_{n+1}=\ol V_n
+he^{-{\ol V_n+\wt V_n(x_{n+1})\/\D T}}\int_0^1g(x-x_{n+1})dx,
\]
\[
\lb{ie2ii}
\wt V_{n+1}(x)=\wt V_n(x)+
he^{-{\ol V_n+\wt V_n(x_{n+1})\/\D T}}\Big(g(x-x_{n+1})
-\int_0^1g(x-x_{n+1})dx\Big),\qq x\in[0,1],
\]
$$
\wt V_n(x)=\ca\wt V_n(0)=\const,&x<0\\\wt V_n(1)=\const,&x>1\ac .
$$

Introduce the new time scale
$\tau(t_n)=h\sum_{j=0}^{n-1}e^{-{\ol V_j\/\D T}}$.
Eq.~\er{ie2ii} takes the form
\[
\lb{ie3ii}
{\wt V_{n+1}(x)-\wt V_n(x)\/\tau(t_{n+1})-\tau(t_n)}
=e^{-{\wt V_n(x_{n+1})\/\D T}}
\Big(g(x-x_{n+1})-\int_0^1g(x-x_{n+1})dx\Big),\qq n\ge 0,\ \ x\in\R.
\]
Introduce the continuous in $x$ and $\tau$ and monotonous in $\tau$ function
$V(x,\tau)$
such that $V(x,\tau(t_n))=V_n(x)$ for all $x\in\R$ and for all
$n\in\N$ large enough. Then
\[
\lb{V=const}
V(x,\tau)=\ca V(0,\tau),& x<0\\
V(1,\tau),&x>1\ac.
\]
Introduce the functions
\[
\lb{wtcV}
\ol V(\tau)=\int_0^1V(x,\tau)dx,\qq
\wt V(x,\tau)=V(x,\tau)-\ol V(\tau).
\]
The function $\wt V(x,\tau)$ is continuous in $x$ and $\tau$ and
satisfies \er{wtVext}.
The density $ p_b(x,\tau)$ of states of the system at the time $\tau\gg 1$
has the form
$$
 p_b(x,\tau)={ e^{-{F(x)+V(x,\tau)\/T}}\/\int_\R e^{-{F(x)+V(x,\tau)\/T}}dx}
={ e^{-{F(x)+\wt V(x,\tau)\/T}}\/Z_b(\tau)},\qq
Z_b(\tau)=\int_\R e^{-{F(x)+\wt V(x,\tau)\/T}}dx.
$$
Instead of \er{defi}, we obtain \er{defiii}.

\subsection{Local solvability}

We prove the local solvability of Eq.~\er{eqcVfiipr}.
Rewrite Eq.~\er{eqcVfiipr} in the form
\[
\lb{ideq2i}
{d\cV(x,\theta )\/d\theta }=\int_\R K(x,s)\gp(s,\theta)ds,
\]
where
\[
\lb{defKrhoi}
K(x,s)=g(x-s)-\int_0^1g(x-u)du,
\qq
\gp(s,\theta;\cV)
=\frac{Y(s,\theta,\cV)}{\cZ_w(\theta,\cV)},
\]
\[
\lb{defYZi}
Y(s,\theta,\cV)=e^{-\a\cV(s,\theta )-\b F(s)},
\qq \cZ_w(\theta,\cV)=\int_\R e^{-\a\cV(s,\theta )-\b F(s)}ds.
\]

\begin{lemma}
\lb{Thlocsolfii}
Let $F\in\cQ$ and let $\cV_0\in\cX$,
where $\cQ$ and $\cX$ are given by \er{defcQF} and \er{defcX}.
Let $M>0$ and let $\cT\in(0,\min\{{M\/2 A_1},{1\/2 A_2}\})$, where
\[
\lb{defmAi}
 A_1={ B_2\/ B_1},\qq
 A_2={\a B_2\/ B_1}( B_2+1),
\]
\[
\lb{defmA1i}
 B_1=\int_\R e^{-\b F(s)}ds,\qq B_2=e^{2\a(\|\cV_0\|+M)}.
\]
Then Eq.~\er{eqcVfiipr} has the unique solution
$\cV\in C^1([0,\cT],\cX)$,
satisfying the conditions

1) $\max_{\tau\in[0,\cT]}\|\cV(\cdot,\tau)-\cV_0\|\le M$,

2) $\cV(x,0)=\cV_0(x),x\in\R$.

Moreover,
\[
\lb{estLc1}
\sup_{x\in\R}\Big|\int_\R K(x-s)\big(\gp(s,\a;\cV_1)-\gp(s,\a;\cV_2)\big)ds\Big|
\le2 A_2.
\]
\end{lemma}

\no {\bf Proof.}
Rewrite Eq.~\er{ideq2i} in the form of the integral equation
\[
\lb{ieqvpci}
\cV(x,\theta )=(\Phi \cV)(x,\theta),\qq
(x,\theta)\in\R\ts[0,\cT],
\]
where the operator $\Phi$, given by
$$
(\Phi \cV)(x,\theta)
=\cV_0(x)+\int_0^\theta\int_\R K(x-s)\gp(s,\a)dsd\a,
$$
acts in the norm space
$
\cM= C([0,\cT],\cX),
$
with the norm
$$
\|\cV\|_{\cM}=\max_{\theta\in[0,\cT]}\|\cV(\cdot,\theta)\|.
$$
For all functions $\cV_1$ and $\cV_2$ from the space $\cM$
and for all $(x,\theta)\in\R\ts[0,\cT]$ we have
$$
(\Phi \cV_1)(x,\theta)-(\Phi \cV_2)(x,\theta)
=\int_0^\theta\int_\R K(x-s)
\big(\gp(s,\a;\cV_1)-\gp(s,\a;\cV_2)\big)dsd\a.
$$
This identity and the estimate $|\int_0^1 K(x-s)ds|\le 2$ give
\[
\lb{estwtVn-V0i}
\|(\Phi \cV)(\cdot,\theta )-\cV_0(\cdot,\theta )\|
=2\int_0^\theta\max_{s\in[0,1]}
\big|\gp(s,\a;\cV)\big|d\a,
\]
\[
\lb{estwtVn-Vn-1i}
\|(\Phi \cV_1)(\cdot,\theta )-(\Phi \cV_2)(\cdot,\theta )\|
=2\int_0^\theta\max_{s\in[0,1]}
\big|\gp(s,\a;\cV_1)-\gp(s,\a;\cV_2)\big|d\a,
\]
for all $\theta\in[0,\cT]$.
Consider the closed set
$$
\cB=\{\cV\in \cM:
\|\cV-\cV_0\|_{\cM}\le M\}
$$
in the space $\cM$.
For all $(s,\theta)\in\R\ts[0,\cT]$ we have from \er{defYZi}
\[
\lb{estZtaui}
{ B_1\/ B_2^{1\/2}}
\le \cZ_w(\theta,\cV)
\le
 B_1 B_2^{1\/2},\qq
Y(s,\theta;\cV)\le h B_2^{1\/2},
\]
for all $\cV\in\cB$,
$ B_1, B_2$ are given by \er{defmA1i},
$$
|Y(s,\theta;\cV_1)-Y(s,\theta;\cV_2)|\le
\Big|e^{-\b F(s)}\big(e^{-\a\cV_1(s,\theta)}
-e^{-\a\cV_2(s,\theta)}\big)\Big|
\le \a B_2^{1\/2}|\cV_1(s,\theta )-\cV_2(s,\theta )|,
$$
for all $\cV_1,\cV_2\in\cB$,
here we used the estimate $|e^{x}-e^{y}|\le e^{\max\{|x|,|y|\}}|x-y|$.
The estimate \er{estZtauq} gives
$$
\begin{aligned}
\Big|{1\/\cZ_w(\theta,\cV_1)}-{1\/\cZ_w(\theta,\cV_2)}\Big|
=\Big|{\cZ_w(\theta,\cV_2)-\cZ_w(\theta,\cV_1)\/\cZ_w(\theta,\cV_2)\cZ_w(\theta,\cV_1)}\Big|
\le { B_2\/ B_1^2}\Big|\int_\R e^{-\b F(s)}\big(e^{-\a\cV_1(s,\theta)}
-e^{-\a\cV_2(s,\theta)}\big)ds\Big|
\\
\le { B_2\/ B_1}\max_{s\in\R}\big|e^{-\a\cV_1(s,\theta)}
-e^{-\a\cV_2(s,\theta)}\big|
\le{\a B_2^{3\/2}\/ B_1}
\big\|\cV_1(\cdot,\theta)-\cV_2(\cdot,\theta)\big\|.
\end{aligned}
$$
Then we obtain from \er{defKrhoi}
$$
\gp(s,\theta;\cV)
\le\frac{|Y(s,\theta,\cV)|}{|\cZ_w(\theta,\cV)|}\le A_1,
$$
$$
\begin{aligned}
\big|\gp(s,\theta;\cV_1)-\gp(s,\theta;\cV_2)\big|
\le
Y(s,\theta;\cV_1)\Big|{1\/Z_b(\theta,\cV_1)}-{1\/Z_b(\theta,\cV_2)}\Big|
+{|Y(s,\theta;\cV_1)-Y(s,\theta;\cV_2)|\/Z_b(\theta,\cV_2)}
\\
\le  A_2\big\|\cV_1(\cdot,\theta)-\cV_2(\cdot,\theta)\big\|.
\end{aligned}
$$
where $ A_1, A_2$ are given by \er{defmAi}.
Substituting these estimates into \er{estwtVn-V0i} and
 \er{estwtVn-Vn-1i}, we obtain
$$
\|\Phi \cV-\cV_0\|_{\cM}
\le 2\cT A_1.
$$
This estimate show that if $\cT<{M\/2 A_1}$, then $\Phi$
maps the set $\cB$ into itself. Moreover, we have
$$
\|(\Phi\cV_1)(\cdot,\theta )-(\Phi\cV_2)(\cdot,\theta )\|
\le 2 A_2\int_0^\theta\|\cV_1(\cdot,\a )-\cV_2(\cdot,\a )\|d\a,
$$
which yields \er{estLc1} and
$$
\|\Phi\cV_1-\Phi\cV_2\|_{\cM}
\le 2 A_2 \cT\|\cV_1-\cV_2\|_{\cM},
$$
If $\cT<\min\{{M\/2 A_1},{1\/2 A_2}\}$, then the mapping $\Phi$ is
contractive on the set $\cB$. Therefore, the integral equation
\er{ieqvpci} has a unique solution in $\cB$.
This solution is differentiable with respect to $\theta$ function.
Therefore, this is a solution to Eq.~\er{eqcVfiipr}
on $\R\ts[0,\cT]$, satisfying the conditions 1)--2).~\BBox

\subsection{Global solvability}
We will show that a global solution exists for the equation
with respect to the variable $\theta$.

\begin{lemma}
\lb{LmeucVi}
Let $F\in\cQ$,
where $\cQ$ is given by \er{defcQF}. Then

i) For any $\cV_0\in\cX$ there exists the unique solution
$\cV\in C^1([0,\cT_{max});\cX)$
to Eq.~\er{eqcVfiipr} for some $\cT_{max}>0$,
satisfying the condition $\cV|_{\theta=0}=\cV_0$.

ii) The solution  $\cV$, defined by i),
extends uniquely to a solution $\cV\in C^1(\R_+;\cX)$.

\end{lemma}

\no {\bf Proof.} i) Lemma~\ref{Thlocsolfii} gives the result.

ii) Let $\cV(x,\theta)$ be the solution to Eq.~\er{eqcVfipr},
satisfying the condition $\cV(x,0)=\cV_0(x),x\in[0,1]$.
Assume that there exists $T_{max}<\iy$ such that
the solution $\cV(\cdot,\theta )$ on the interval
$\theta\in[0,T_{max}]$, and this solution cannot be extended
onto the larger interval $\theta\in[0,T_{max}+\d]$, for any $\d>0$.

Eq.~\er{eqcVfiipr} gives
$$
\max_{x\in[0,1]}\Big|{d\cV(x,\theta )\/d\theta }\Big|
\le 2\max_{x\in\R}|g(x)|\int_0^1\gp(x,\theta)dx=2\max_{x\in\R}|g(x)|,\qq
\forall\ \ \theta\ge 0.
$$
This yields
\[
\lb{estcVt1-cVt2ii}
\max_{x\in[0,1]}\big|\cV(x,\theta_1)-\cV(x,\theta_2)\big|\le
2\max_{x\in\R}|g(x)||\theta_1-\theta_2|,\qq \forall\ \ \theta_1,\theta_2\ge 0.
\]
Since the space $\cX$ is complete, the estimate \er{estcVt1-cVt2ii}
yields that the limit $\cV_*=\lim_{\theta\to T_{max}}\in \cX$
exists. Lemma~\ref{Thlocsolfii}
implies that there exists the unique solution $\cV(\cdot,\theta)$,
$\theta\in[0,T_{max}+\d]$, for some $\d>0$, satisfying the condition
$\cV(\cdot,T_{max})=\cV_*$.
Therefore, the solution $\cV(x,\theta)$ can be extended
uniquely from $[0,T_{max}]$
 onto a larger interval, which gives the result.~\BBox

\medskip

\no {\bf Proof of Theorem~\ref{TheucVi}.}
Lemma~\ref{LmeucVi} gives the result.~\BBox

\subsection{Quasi-stationary solution}
Introduce the class $\mP$ of probability measure densities corresponding
to the class $\cX$ of biased potentials:
\[
\lb{spmP}
\begin{aligned}
	\mP=\{p\in C(\R):p>0,\int_\R p(x)dx=1;p(x)=p(0)e^{\b(F(0)-F(x))},\text{ as }x<0;
	\\
	p(x)=p(1)e^{\b(F(1)-F(x))},\text{ as }x>1\}.
\end{aligned}
\]
Let $\cV\in\cX$, then
$p={1\/Z_w}e^{-\a\cV-\b F}\in\mP$. Assume, in addition,
that $F(x)=F(0)$ for all $x\in[-L,0]$ and $F(x)=F(1)$ for all $x\in[1,L+1]$
for some $L>0$ large enough. Introduce the density $p_L\in\mP$
\[
\lb{defqstd}
p_L(x)={1\/Z_L}\ca
e^{\b(F(0)-F(x))}, x<-L,\\
\qq 1,\qq -L\le x\le L+1,\\
e^{\b(F(1)-F(x))}, x>L+1,
\ac
\]
$Z_L$ is given by \er{ZL}.
We show that $p_L$ yields a quasistationary distribution.
Let $\cV_L$ denote the biased potential \er{defcVL}
corresponding to the distribution $p_L$ via formula \er{pcV}.
In the set $\cX$, consider a ball \er{bincX} centered at $\cV_L$.

\begin{theorem}
\lb{ThQSSg}
 Let $\cV$ be a solution to Eq.~\er{eqcVfiipr}
satisfying the condition $\cV(\cdot,0)=\cV_L$, and let $\theta>0$
be such that $\cV(\cdot,\theta)\in\cB_R$ for some $R>0$.
Then
\[
\lb{estcV-cV0}
\max_{x\in[0,1]}|\cV(x,\theta )-\cV_L(x)|\le{\ve_L\/K_R}(e^{K_R\theta}-1),
\]
where $K_R$ is given by \er{defkRLc}, $\ve_L$ has the form \er{defveL}.
Moreover, for $\theta$, satisfying \er{tstgtVinBR},
the solution  $\cV$ remains in the ball $\cB_R$.

\end{theorem}

\no {\bf Proof.}
Let $a_L(x)=g\star p_L(x)$. Then the definition \er{Aisconvi} gives
$$
\begin{aligned}
a_L(x)=\int_\R g(x-s)p_L(s)ds={1\/Z_L}
\Big(\int_{-\iy}^{-L}g(x-s)e^{\b(F(0)-F(s))}ds
+\int_{-L}^{L+1}g(x-s)ds
\\
+\int_{L+1}^\iy g(x-s)e^{\b(F(1)-F(s))}ds\Big).
\end{aligned}
$$
Using the estimates
$$
\max_{x\in[0,1],s\le -L}|g'(x-s)|=-g'(L)={Le^{-{L^2\/2\s^2}}\/\sqrt{2\pi}\s^3},
\qq
\max_{x\in[0,1],s\ge L+1}|g'(x-s)|=g'(-L)=-g'(L),
$$
we obtain
$$
\begin{aligned}
\max_{x\in[0,1]}\Big|\int_{-L}^{L+1}g'(x-s)ds\Big|
=\max_{x\in[0,1]}|g(x-L-1)-g(x+L)|
\\
\le{1\/\sqrt{2\pi}\s}\max_{x\in[0,1]}\Big(e^{-{(x-L-1)^2\/2\s^2}}
+e^{-{(x+L)^2\/2\s^2}}\Big)
 \le {2e^{-{L^2\/2\s^2}}\/\sqrt{2\pi}\s},
\end{aligned}
$$
$$
\begin{aligned}
\max_{x\in[0,1]}\Big|\int_{-\iy}^{-L}g'(x-s)e^{\b(F(0)-F(s))}ds
+\int_{L+1}^\iy g'(x-s)e^{\b(F(1)-F(s))}ds\Big|
\\
\le{Le^{-{L^2\/2\s^2}}\/\sqrt{2\pi}\s^3}
\Big(\int_{-\iy}^{-L}e^{\b(F(0)-F(s))}ds
+\int_{L+1}^\iy e^{\b(F(1)-F(s))}ds\Big),
\end{aligned}
$$
which yields
$$
\max_{x\in[0,1]}|a_L'(x)|\le\ve_L.
$$
This gives
\[
\lb{cA=cA-cA}
\max_{x\in[0,1]}|a_L(x)-\langle a_L\rangle|\le
\max_{x\in[0,1]}a_L(x)-\min_{x\in[0,1]}a_L(x)\le
\max_{x\in[0,1]}|a_L'(x)|\le\ve_L.
\]
Then
$$
{d(\cV(x,\theta )-\cV_L(x))\/d\theta }
={d\cV(x,\theta )\/d\theta }=(G\cV)(x,\theta)
=(G\cV)(x,\theta)-(G\cV_L)(x)+(G\cV_L)(x),
$$
as $(x,\theta)\in[0,1]\ts\R_+$, where
$$
(G\cV)(x,\theta)=\int_\R\Big(g(x-s)-\int_0^1g(u-s)du\Big)\gp(s,\theta)ds.
$$
The estimate \er{cA=cA-cA} gives $\max_{x\in[0,1]}|(G\cV_L)(x)|<\ve_L$,
then, using \er{estLc1}, we obtain
$$
D_\theta^+\max_{x\in[0,1]}|\cV(x,\theta )-\cV_L(x)|
\le\max_{x\in[0,1]}\Big|{d(\cV(x,\theta )-\cV_L(x))\/d\theta }\Big|
\le K_R\max_{x\in[0,1]}|\cV(x,\theta)-\cV_L(x)|+\ve_L,
$$
for all $\theta>0$ such that $\cV(\cdot,\theta )$ remains in the ball $\cB_R$, here
$D^+f(x)=\lim\sup_{h\to+0}{f(x+h)-f(x)\/h}$ is the upper right Dini derivative.
Applying Gronwall's lemma, we obtain \er{estcV-cV0}.

Let $\theta_{R,L}$ denote the time at which the solution $\cV$
reaches the boundary of the ball
$\cB_R$. The estimate \er{estcV-cV0} gives
$$
R\le{\ve_L\/K_R}(e^{K_R\theta_{R,L}}-1),
$$
which yields the estimate
$$
\theta_{R,L}\ge {1\/K_R}\log\Big({K_R\/\ve_L}R+1\Big).
$$
This gives the estimate \er{tstgtVinBR}.~\BBox

\medskip

\no {\bf Proof of Theorem~\ref{ThQSS}}.
The result follows from Theorem~\ref{ThQSSg}.~\BBox

\subsection{Energy functional}
The identity
\[
\lb{pcV}
p(x)={e^{-\a\cV(x)-\b F(x)}\/\int_\R e^{-\a\cV(s)-\b F(s)}ds},\qq
x\in\R,
\]
establishes an one-to-one correspondence between the sets $p\in\mP$ and
$\cV\in\cX$.

Rewrite Eq.~\er{eqcVfiipr} in the form of equation for $\gp$.

\begin{lemma}
i) The function $\gp$ satisfies the equation
\[
\lb{eqgpwti}
{d\gp(x,\theta)\/d\theta}
=-\a\gp(x,\theta)
\big(g\star \gp(x,\theta)-\langle g\star \gp\rangle_\gp(\theta)\big),\qq x\in[0,1],
\]
where
$$
\langle f\rangle_\gp=\int_\R f(x)\gp(x)dx.
$$

ii) Let $\gp_0\in\mP$.
Then the solution $\gp(\cdot,\theta)\in\mP$
to Eq.~\er{eqgpwt}, satisfying the condition
$\gp(x,0)=\gp_0(x),x\in[0,1]$,
exists and is unique for all $\theta>0$.

\end{lemma}

\no {\bf Proof.} The proof repeats the arguments from the proof of
Lemma~\ref{Lmglclsol}.

i) Differentiating the expression for $p_w$ from \er{defgpZ1i}
and substituting the result into \er{defgpg*fi}, we obtain
\[
\lb{difgpwti}
{d\gp(x,\theta)\/d\theta}
=-\a\gp(x,\theta){d\cV(x,\theta )\/d\theta }
-{\gp(x,\theta)\/Z_w(\tau(\theta))}{dZ_w(\tau(\theta) )\/d\theta }
=-\a\gp(x,\theta)\Big({d\cV(x,\theta )\/d\theta }
-\Big\langle{d\cV\/d\theta }\Big\rangle_\gp(\theta )\Big).
\]
The identity \er{eqcVfiipr} gives
$$
\Big\langle{d\cV\/d\theta }\Big\rangle_\gp(\theta )
=\int_\R\big( g\star \gp(x,\theta)
-\langle g\star \gp\rangle(\theta)\big)\gp(x,\theta )dx
=\langle g\star \gp\rangle_\gp(\theta)-\langle g\star \gp\rangle(\theta),
$$
then
$$
{d\cV(x,\theta )\/d\theta }-\Big\langle{d\cV\/d\theta }\Big\rangle_\gp(\theta )
=g\star \gp(x,\theta)-\langle g\star \gp\rangle_\gp(\theta).
$$
Substituting this into \er{difgpwti}, we obtain equation \er{eqgpwti}.

ii) The result follows from Lemma~\ref{LmeucVi}~ii).~\BBox

\medskip

Define a positive-definite functional $\cE$ on the set $\mP$ by the identity
\[
\lb{defcEi}
\cE[ p ]={1\/2}\int_\R\int_\R g(x-y) p (x) p (y)dxdy.
\]
The functional $\cE$ possesses all the properties
of a Lyapunov function for Eq.~\er{eqgpwti}.
We introduce the functional $\cD$ on $\mP$ by the identity
\[
\lb{defcDi}
\cD[p]=\int_\R p(x)\big(g\star p(x)-\langle g\star p\rangle_p\big)^2dx
=\langle (g\star p)^2\rangle_p-\langle g\star p\rangle_p^2,
\]
$\cD$ is the variance of $g \star p$ with respect to the measure $p$.
It follows from the definition that $\cD[p] \ge 0$ for all $p \in \mP$.

\begin{lemma}
Let $\gp$ be a solution to Eq.~\er{eqgpwti}. Then
\[
\lb{derLFi}
{d\cE[\gp (\cdot,\theta)]\/d\theta}
=-\a\cD[\gp(\cdot,\theta)]\le 0.
\]
In particular, the function $\cE[\gp (\cdot,\theta)]$ is monotonically
decreasing and has a non-negative limit as $\theta\to\iy$.
The functional $\cD$ satisfies
\[
\lb{varcDi}
\int_0^\iy\cD[\gp(\cdot,\theta)]d\tau<\iy.
\]
\end{lemma}

\no {\bf Proof.} By repeating the arguments from
the proof of Lemma~\ref{LmLF}, we obtain the desired results.~\BBox

\medskip

\no {\bf Remark.}
It follows from this lemma that the system \er{eqgpwti},
just like the system \er{eqgpwt},
cannot exhibit complex dynamics; see Proposition~\ref{Lmdiss}.

\subsection{Week solution}
Let $F\in\cQ$ be fix.
Introduce the space $\cM_F$ of probability measures of the form
\[
\lb{mcMF}
\m(dx)=\eta(dx)+c_-e^{-\b(F(x)-F(0))}\1_{(-\iy,0)}dx+
c_+e^{-\b(F(x)-F(1))}\1_{(1,+\iy)}dx,
\]
where $\eta\in\cM[0,1]$, $c_-\ge 0,c_+\ge 0$, satisfying the condition
\[
\lb{nmcMF}
\m(\R)=\eta([0,1])+c_-\int_{-\iy}^0e^{-\b(F(x)-F(0))}dx+
c_+\int_1^{+\iy}e^{-\b(F(x)-F(1))}dx=1.
\]
The set $\cM_F$ is closed with respect to the weak-* convergence
$$
\m_n\rightharpoonup\m\Leftrightarrow
\int_\R\vp(x)\m_n(dx)\rightarrow\int_\R\vp(x)\m(dx)\qq
\forall\ \ \vp\in C_b(\R),
$$
where $C_b(\R)$ is the space of continuous functions bounded on $\R$.
Moreover, $\cM_F$ is the weak-* closure of the space of measures with
densities from the class $\mP$ defined by \er{spmP}.

For $\m\in\cM_F$ we define the Gaussian potential
\[
\lb{gpi}
U_{\m}(x)=\int_{\R}g(x-s)\m(ds),
\]
its average with respect to the measure
$$
\langle U_\m\rangle_\m
=\int_{\R}U_{\m}(x)\m(dx),
$$
the energy functional
\[
\lb{defcEfii}
\cE[\m]={1\/2}\int_{\R^2}g(x-y)\m(dx)\m(dy)
={1\/2}\int_{\R}U_\m(x)\m(dx),
\]
and the dissipation
\[
\lb{defcDfi}
\cD[\m]=\int_\R\big(U_\m(x)-\langle U_\m\rangle_{\m}\big)^2\m(dx).
\]

\begin{lemma}
\lb{Lmgpi}

i) The mapping $\m\mapsto U_\m$ is a continuous mapping from $\cM_F$
to $C[0,1]$; that is, weak-* convergence $\m_n\to\m$ implies
$\max_{[0,1]}|U_\m(x)-U_{\m_n}(x)|\to 0$.

ii) The functionals $\cE$ and $\cD$ are continuous in the metric  $\cM_F$;
that is, weak-* convergence $\m_n\to\m$ implies $\cE[\m_n]\to\cE[\m]$
 and $\cD[\m_n]\to\cD[\m]$. Moreover, if a family of measures
$\m_\theta^{(n)}$ converges in $\cM_F$ to a family $\m_\theta$
uniformly in $\theta$ on a compact set $[0,\cT]$ (where $\cT>0$),
then $\cE[\m_\theta^{(n)}]$ and $\cD[\m_\theta^{(n)}]$ converge to
$\cE[\m_\theta]$ and
$\cD[\m_\theta]$, respectively, uniformly on that compact set.

\end{lemma}

\no {\bf Proof.} Repeating the arguments from the proof of
Lemma~\ref{Lmgpi}, we obtain the desired result.~\BBox

\medskip

In order to rewrite equation \er{eqgpwti} in weak form,
we introduce a modified Gaussian potential
\[
\lb{modpot}
H_\m(x)=\ca U_\m(0),x<0\\U_\m(x),0\le x\le 1\\ U_\m(1),x>1\ac.
\]
By a weak solution of the equation
\[
\lb{eqgpwtffi}
{d\m_\theta\/d\theta}
=-\a\big(U_{\m_\theta}(x)-\langle U_{\m_\theta}\rangle_{\m_\theta}\big)\m_\theta,
\qq \a>0,
\]
we mean a family of probability measures $\m_\theta\in\cM_F$ such that,
for any $\vp\in C_b(\R)$ the function $t\mapsto\int_{\R}\vp(x)\m_\theta(dx)$
is absolutely continuous and the equation
\[
\lb{eqgpwtfi}
{d\/d\theta}\int_{\R}\vp(x)\m_\theta(dx)
=-\a\int_{\R}\vp(x)
\big(H_{\m_\theta}(x)-\langle H_{\m_\theta}\rangle_{\m_\theta}\big)\m_\theta(dx)
\]
holds for almost all $\theta$.
Weak stationarity of a measure means that
\[
\lb{steqfifi}
\int_{\R}\vp(x)
\big(H_{\m}(x)-\langle H_{\m}\rangle_{\m}\big)\m(dx)=0
\]
for all $\vp\in C_b(\R)$.
This is equivalent to
\[
\lb{cstmcMF}
H_{\m}(x)=\langle H_{\m}\rangle_{\m}
\]
$\m$-almost everywhere; that is, the measure is weakly stationary
if and only if its potential is constant on the support of the measure.

\begin{lemma}
\lb{Lmfsmdi}
Let $\m\in\cM_F$ be a weak stationary measure of the form \er{mcMF}.
Then $\eta$ is a purely atomic measure with no accumulation points,
that is,
\[
\lb{stmdiskri}
\eta=\sum_{n=1}^Na_n\d_{x_n},\qq a_n>0,\qq 0=x_1<x_2<...<x_N=1,\qq n=1,...,N,
\]
for some $N\in\N$.

\end{lemma}

\no {\bf Proof.}
Let $\m\in\cM_F$ be a weak stationary measure. Then,
according to \er{cstmcMF}, the identity $H_\m(x)=\l$
holds $\m$-almost everywhere, where $H_\m$ is defined by \er{modpot}
and $\l$ is a real constant. On the interval $[0,1]$,
we have $H_\m(x)=U_\m(x)$.
Since $U_\m$ admits an analytic continuation to the entire complex plane,
the equation $U_\m(x)=\l$ has only a finite number
of solutions on the interval $[0,1]$.
It follows that $\m$ is a finitely atomic measure on the interval $[0,1]$.~\BBox

\subsection{Dynamics of a weak solution}

As in the case of a finite interval, the system is purely dissipative,
exhibits no stable cycles, and its long-term behavior is entirely
determined by the $\o$-limit set. The trajectory cannot
possess non-trivial periodic orbits, the $\o$-limit set consists
of weakly stationary measures, and the asymptotic dynamics
are non-oscillatory: the solution does not wander
between different regimes but relaxes to one of
the limit sets of weakly stationary measures with constant energy.

\begin{theorem}
\lb{ThsIws}
Let $\m_0\in\cM_F$. Then there exists a global weak solution
$(\m_\theta)_{\theta\ge 0}\ss\cM_F$
to Eq.~\er{eqgpwtffi} satisfying the initial condition
$\m_\theta|_{\theta=0}=\m_0$. Furthermore,

i) The function $\theta\mapsto\cE[\m_{\theta}]$, given by \er{defcEfii},
satisfies the identity
\[
\lb{derLFfi}
\cE[\m_t]+\a\int_s^t\cD[\m_u]du
=\cE[\m_s],\qq 0\le s\le t,
\]
where $\cD$ is defined by \er{defcDfi}.
In particular, the function $\theta\mapsto\cE[\m_{\theta}]$
is non-increasing and has a finite limit.

ii) The function $\theta\mapsto\cD[\m_{\theta}]$, given by \er{defcDf},
satisfies the estimate
\[
\lb{varcDf3}
\int_0^\iy\cD[\m_{\theta}]d\theta<\iy.
\]
In particular, $\theta\mapsto\cD[\m_{\theta}]$ tends to zero,
at least along subsequences $\theta_n\to\iy$.

iii) Let $\o(\m_0)$ denote the set of limit points of
the trajectory $\m_{\theta}$ starting at $\mu_0$;
$\o(\m_0)$ consists of points $\m_\iy$ of the set $\cM[0,1]$ such that
$\m_{t_n}\rightharpoonup\m_\iy$
along some sequence $t_n\to+\iy$. Then

a) any point $\m_\iy$ of the set $\o(\m_0)$
is a weakly stationary measure:
\[
\lb{lmfs3}
\big(U_{\m_\iy}(x)-\langle U_{\m_\iy}\rangle_{\m_\iy}\big)\m_\iy=0,
\]

b) the set $\o(\m_0)$ consists of finite-atomic measures,

c) $\cE$ is constant on the set $\o(\m_0)$.

\end{theorem}

\no {\bf Proof.} By repeating the arguments from the proof
of Theorem~\ref{Thfs},
we obtain the desired results.~\BBox

\section{Numerical calculations}

\subsection{Standard and well-tempered metadynamics in the adiabatic limit}
Metadynamics works well in the adiabatic limit, which means that the system reaches equilibrium rapidly with respect to changes in the bias potential. As a result, the Gaussian hills are deposited according to the instantaneous equilibrium density in the total potential $F(x)+V(x)$ \cite{DPV14}. In order to illustrate the analysis carried out in this paper, we have implemented a stochastic numerical approach that we refer to as {\it adiabatic metadynamics in the instantaneous-equilibrium limit}. In this approach, we explicitly impose the adiabatic limit in its strongest sense: after each update of the bias, the next deposition point is drawn directly from the instantaneous biased equilibrium distribution. This effectively reduces the relaxation time to zero, removing the effects associated with a finite relaxation time and thus isolating the inherent evolution of the metadynamics bias. The code of the package can be found on GitHub \cite{R26}.

The ODE equation describing the evolution of the bias potential \er{dpv} was derived only for well-tempered metadynamics \cite{DPV14}. Therefore, the transition to the limit $\Delta T\to\infty$, which leads to a similar equation for standard metadynamics, could only be performed formally. Nevertheless, this equation can be derived independently for standard metadynamics.  The intuitive idea of the proof is as follows:

Let $\Gamma_0$ be the zero-mean part of the hill, $h>0$ be the constant hill height, $X_{n+1}$ be the center of the hill deposited at step $n+1$, and  $\widetilde V_n$ be the zero-mean bias after $n$ depositions, then:
\[
\label{p1}
\widetilde V_{n+1}(x)-\widetilde V_n(x)
=h\Gamma_0(x,X_{n+1}).
\]

We denote by $\mathcal F_n$ the complete history up to step $n$
and define the conditional mean hill $\mathcal A[\widetilde V_n](x)$
and fluctuation produce by the random position of the hill
$\xi_{n+1}(x)$.
\[
\label{p2}
\mathcal A[\widetilde V_n](x)
:=\mathbb E\!\left[\Gamma_0(x,X_{n+1})\mid\mathcal F_n\right].
\]
\[
\label{p3}
\xi_{n+1}(x)
:=\Gamma_0(x,X_{n+1})-\mathcal A[\widetilde V_n](x).
\]
The exact discrete update becomes
\[
\label{p4}
\widetilde V_{n+1}(x)-\widetilde V_n(x)
=h\mathcal A[\widetilde V_n](x)+h\xi_{n+1}(x),
\]
where the first term is the deterministic conditional drift.
Summing \er{p4} from $0$ to $n-1$ gives
\[
\label{p5}
\widetilde V_n
=\widetilde V_0
+h\sum_{j=0}^{n-1}\mathcal A[\widetilde V_j]
+h\sum_{j=1}^{n}\xi_j.
\]

Assuming the adiabatic approximation is valid and conditioned on
the $\mathcal F_n$, the next center $S_{n+1}$ is distributed according
to the equilibrium distribution corresponding to the current bias:
\[
\label{p6}
\mathcal A[\widetilde V_n](x)
=\int \Gamma_0(x,x')p_b(x';\widetilde V_n)\,dx'.
\]
By introducing the internal time $\tau = nh$, and transitioning
to continuous internal time, in the small-hill limit \er{p5} becomes:
\begin{equation}
    \widetilde V(\tau)
    =
    \widetilde V(0)
    +
    \int_0^\tau
    \mathcal A[\widetilde V(\tau')]\,d\tau'.
\label{p7}
\end{equation}
The accumulated fluctuation term has typical magnitude $O(\sqrt h)$. Hence it vanishes as $h\to0$.

If $\mathcal A$ is continuous, the limiting trajectory is
differentiable in $\tau$ and satisfies
\begin{equation}
        \frac{\partial\widetilde V(x,\tau)}{\partial\tau}
        =
        \int
        \Gamma_0(x,x')
        p_b(x';\widetilde V(\tau))\,dx'.
\label{p8}
\end{equation}
Consequently, the evolution of bias potential in standard metadynamics in the small-hill limit is characterized by the same ODE equation as well-tempered metadynamics.
This conclusion is significant for our numerical illustration because the internal time of well-tempered metadynamics slows down logarithmically with respect to physical time. This makes tracking the long-term evolution of the bias challenging. In standard metadynamics, the internal time is a linear function of the physical time; thus, the long-term evolution of the bias potential can be tracked within a reasonably short simulation. Notably, in internal metadynamics time, both metadynamics modifications result in the same conditional mean drift of the bias potential up to the last internal time reached by well-tempered metadynamics (see Figures 4a and 4b). This observation supports our hypothesis that standard metadynamics in the small-hill limit and well-tempered metadynamics in the long-time limit can be described by the same ordinary differential equation (ODE). We have included both metadynamics modifications in our package.

\subsection{Trajectory-based metadynamics}
In addition to metadynamics in the instantaneous-equilibrium limit, we also performed trajectory-based metadynamics. In these simulations, the centres of the hills are not sampled from the equilibrium distribution, but rather from the actual trajectory obtained by integrating the equations of motion. To this end, we employed the PLUMED PesMD engine \cite{PLMD14}, \cite{PLMD19} which supports Langevin dynamics. The following parameters, given in reduced units, were common to all simulations: temperature and the Boltzmann constant were set to 1.0 and the Langevin damping rate was set to 0.1.

\end{document}